\documentclass[sn-apa]{sn-jnl}

\usepackage{graphicx}
\usepackage{multirow}
\usepackage{amsmath,amssymb,amsfonts}
\usepackage{mathrsfs}
\usepackage{amsthm}
\usepackage{booktabs}
\usepackage{appendix}
\usepackage{xcolor}
\usepackage{textcomp}
\usepackage{manyfoot}

\usepackage{algorithm}
\usepackage{algorithmicx}
\usepackage{algpseudocode}
\usepackage{listings}

\usepackage[none]{hyphenat}
\usepackage{enumerate}
\usepackage{enumitem}
\usepackage{setspace}
\usepackage{url}
\usepackage{cleveref}
\usepackage{caption}
\usepackage{float}
\usepackage{lmodern}
\usepackage{adjustbox}
\usepackage{subcaption}

\renewcommand{\S}{\mathcal{S}}

\renewcommand{\d}{\mathrm{d}}

\newtheorem{lem}{Lemma}
\renewcommand\thesection{\arabic{section}}
\renewcommand\thesubsection{\thesection.\arabic{subsection}}

\begin{document}
\sloppy{}
\title[Article Title]{Adapting Pairs Trading to Gambling Markets: A Case Study of the U.S.\ Presidential Election
}


\author*{\fnm{Haoyu} \sur{Liu}}\email{hl215@st-andrews.ac.uk}

\author{\fnm{Len} \sur{Thomas}}\email{len.thomas@st-andrews.ac.uk}
\author{\fnm{Benjamin} \sur{Baer}}\email{benjamin.baer@st-andrews.ac.uk}
\author{\fnm{Carl} \sur{Donovan}}\email{crd2@st-andrews.ac.uk}

\affil{\orgdiv{School of Mathematics and Statistics}, \orgname{University of St Andrews}, \orgaddress{\city{St Andrews}, \postcode{KY16 9SS}, \state{Fife}, \country{United Kingdom}}}
\abstract{ Pairs trading exploits mean reversion in the relationship between related assets. We adapt this idea to political betting markets by modelling the combined implied probability of the two major-party nominees with a latent Ornstein-Uhlenbeck process whose mean-reversion level varies over time and whose observations contain additive noise. Model parameters are estimated from regularly sampled odds data using a state-space likelihood, with consecutive repeated values represented by a single retained observation and the elapsed number of sampling intervals preserved in the continuous-time transition. Parametric-bootstrap upper prediction bounds identify signal times at which the combined implied probability is likely to decline, and a no-intercept Bradley-Terry-type model selects the candidate-specific odds quote. The candidate-selection model is trained on 2020 U.S.\ presidential-election data and evaluated out of sample on 2024 data. The 2024 analysis produced 130 signals, empirical one-step coverage of 95.1\%, a mean synthetic odds-price return of 1.86\%, and an unannualized per-trade Sharpe-type ratio of 1.12. These returns are frictionless descriptive quantities rather than executable betting-exchange profits. The results support the integrated framework as a proof of concept for two-candidate electoral markets.}

\keywords{Political Betting Markets; Trending Ornstein-Uhlenbeck Process; Parametric Bootstrap; Bradley-Terry Model}


\maketitle

\newpage

\section{Introduction}

Pairs trading is a common financial strategy that exploits temporary deviations from the historical relationship between two correlated assets \citep{elliott2005pairs}. The strategy involves identifying two related tradable instruments whose prices typically move together and taking opposite positions when their relative prices diverge, with the expectation that the spread will eventually revert towards its historical equilibrium. For example, if the price of one stock increases while that of a related stock remains unchanged, a trader may short the rising stock and buy the lagging stock in anticipation of subsequent convergence. The resulting long-short position can be represented by a spread, typically defined as a weighted difference between the prices of the two assets. When the weighting, or hedge ratio, is appropriately chosen, the portfolio may have limited sensitivity to broad market movements \citep{elliott2005pairs}.

A frequently cited illustration is the relationship between Coca-Cola and PepsiCo shares. When their prices depart from their historical relationship, a trader may buy the relatively underperforming stock and sell the relatively outperforming stock, seeking to profit when the spread converges. The strategy therefore relies on mean reversion: the phenomenon whereby the spread between related assets tends to return towards its long-run level. 

Mean-reverting behaviour may also arise in gambling markets. Because financial and gambling markets share several trading characteristics, some financial-market methods may be adapted to gambling-market settings \citep{liu2026comparison}.

Peer-to-peer (P2P) systems allow participants with similar roles to
interact and exchange resources directly, rather than assigning fixed client and server roles \citep{vansteen2023distributed}. This structure has also influenced the design of electronic markets in which participants trade with one another through a common platform. In betting exchanges, for example, participants submit back and lay orders against other market participants rather than accepting odds set solely by a bookmaker \citep{liu2026comparison}. In betting markets, this structure is implemented through betting exchanges, where participants trade against one another through a common platform rather than betting directly against a bookmaker. Participants can back an outcome or take the opposing position by laying it \citep{casadesus2019platform}.

Political betting markets provide a particularly relevant setting for the adaptation of pairs-trading methods.\ U.S.\ presidential elections are held every four years, with Election Day falling on the first Tuesday after the first Monday in November. Although third-party and independent candidates may also contest these elections, the Democratic and Republican parties have historically dominated presidential contests. Once the two major-party nominees have been confirmed, the market can therefore be treated approximately as a two-candidate contest, providing a natural parallel with pairs trading.

Betting on U.S.\ presidential elections through P2P exchanges often begins more than a year before Election Day. During the election year, candidates are selected through primary elections, followed by national conventions at which the party nominees are formally confirmed. In the 2020 election, Joe Biden secured the Democratic nomination on 18 August and Donald Trump received the Republican nomination on 24 August. In the 2024 election, Trump was nominated by the Republican Party on 15 July, while Kamala Harris accepted the Democratic presidential nomination on 22 August. Accordingly, this study focuses on betting activity between August and November, after the two major-party nominees had been confirmed.

Decimal betting odds can be inverted to obtain implied probabilities of winning. In a fair market, the implied probabilities across a complete set of mutually exclusive and collectively exhaustive outcomes sum to one. In observed betting markets, this sum may differ from one. When it exceeds one, the excess is referred to as the market over-round \citep{koning2023betting}. However, the present study considers only the Democratic and Republican nominees rather than all outcomes listed in the market. Their combined implied probability may therefore lie below one when probability mass is assigned to other candidates or outcomes.

The observed combined implied probabilities exhibit short-run fluctuations around an underlying level that changes gradually over time. This level generally increases as the election approaches, which is consistent with a declining contribution from alternative candidates and motivates a time-varying mean-reversion model.

The Ornstein-Uhlenbeck process is a continuous-time model commonly used to describe mean-reverting behaviour. It is particularly suited to variables that fluctuate around an equilibrium level, such as spreads in pairs-trading strategies \citep{mudchanatongsuk2008optimal}. Although pairs trading and Ornstein-Uhlenbeck models have been widely studied in financial markets, their application to high-frequency political betting exchange data remains limited. This paper addresses this gap by developing a pairs-trading-inspired framework for the U.S.\ presidential election betting market. We model the high-frequency dynamics of the summed implied probabilities using a trending Ornstein-Uhlenbeck (OU) process with additive white noise. We estimate the model parameters using a parametric bootstrap procedure and use the resulting bootstrap distribution to construct one-sided prediction intervals. We identify betting signals when observations exceed the corresponding prediction bounds and then apply an antisymmetric decision rule to determine which candidate to back. We use data from the 2020 U.S.\ presidential election to construct an antisymmetric decision model and evaluate it using out-of-sample data from the 2024 election. Finally, we assess strategy performance using the mean realized return and the Sharpe ratio, which measures mean return relative to return variability \citep{sharpe1998sharpe}.

The contribution of this paper is threefold. First, we adapt the concept of pairs trading from traditional financial markets to a P2P political betting market. Second, we develop a trending OU process with additive white noise to represent both short-term mean reversion and the long-term trend in the combined implied probability. Third, we combine bootstrap-based prediction intervals with an antisymmetric candidate-selection rule to construct a proof-of-concept betting strategy. Although we evaluate the proposed framework using data from the 2020 and 2024 U.S.\ presidential elections, it may also be applicable to future elections in which the betting market remains effectively dominated by two major-party candidates.

\subsection{Data and study design}

We use data from the 2020 U.S.\ presidential election as the training data. The observation period extends from 23 August 2020, immediately before the Republican National Convention and one day before Trump's formal renomination, to 4 November, the end of the observation period. The resulting cleaned dataset contains 21,181 observations. The two candidates considered are Trump of the Republican Party and Biden of the Democratic Party. \Cref{fig:2020-data} shows the implied probabilities of Trump and Biden during the 2020 observation period. The red line represents Trump and the blue line represents Biden. The two series show a clear inverse relationship over much of the sample. An important feature of the data is the period surrounding Trump's COVID-19 diagnosis, which is marked by the orange dashed lines. During part of this interval, both implied probabilities remain nearly unchanged and form a visible flat period. The series resume more substantial movements after this period. This episode illustrates how major political events can be associated with temporary changes in the dynamics of election betting odds.

\begin{figure}[htbp]
    \centering
    \includegraphics[width=5in]{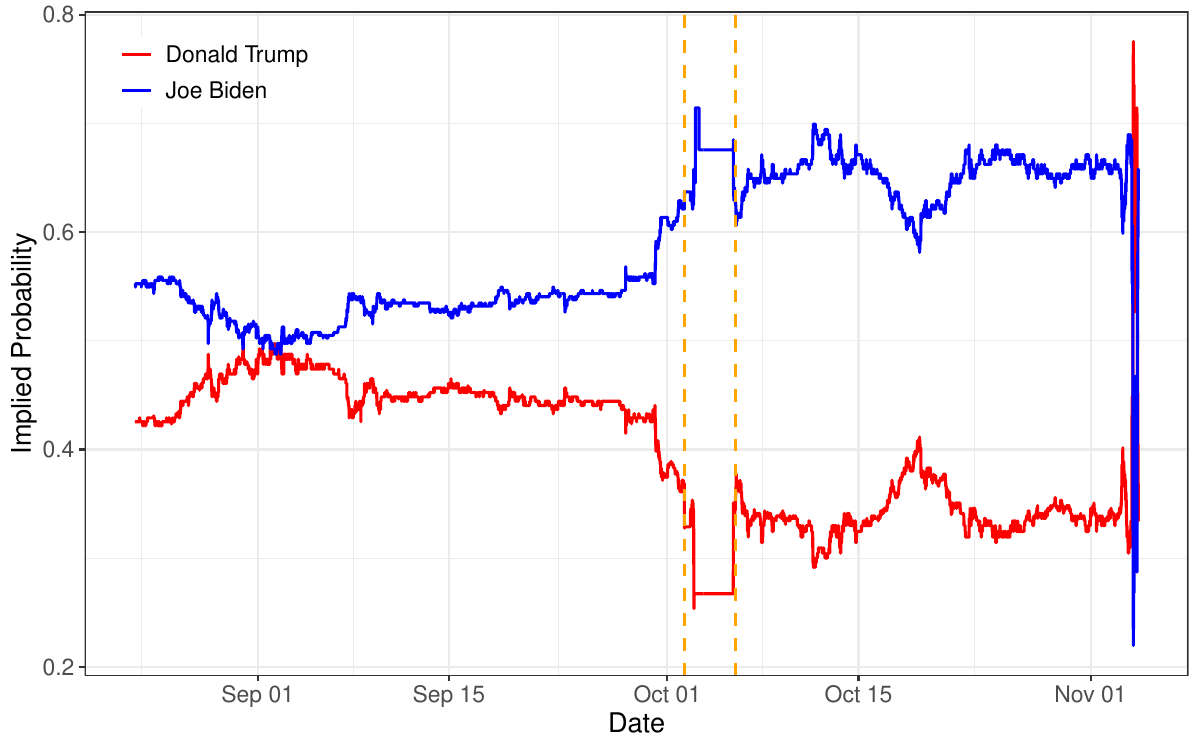}
    \captionsetup{justification=centering}
    \caption{Implied probabilities of Trump and Biden in the 2020 U.S.\ presidential election betting market. The orange dashed lines mark the period of Trump's COVID-19 infection, during which the series contains a visible flat segment.}
    \label{fig:2020-data}
\end{figure}

We obtained the data from Betdata, an independent platform that records political betting odds from the Betfair exchange and maintains an archive of historical political betting markets. Betfair, a UK-based company, operates a major peer-to-peer betting exchange \citep{gonccalves2019deep}. The 2020 observations are used to construct the betting strategy and estimate the candidate-selection model.

For out-of-sample evaluation, we use data from the 2024 U.S.\ presidential election. The observation period runs from 23 August, after the two major-party nominees had been confirmed, to 6 November. The two candidates are Trump of the Republican Party and Harris of the Democratic Party. Odds were recorded at five-minute intervals, yielding 21,689 observations after data cleaning. This sampling frequency captures changes in the market while avoiding excessive short-term variation.

\begin{figure}[htbp]
    \centering
    \includegraphics[width=5in]{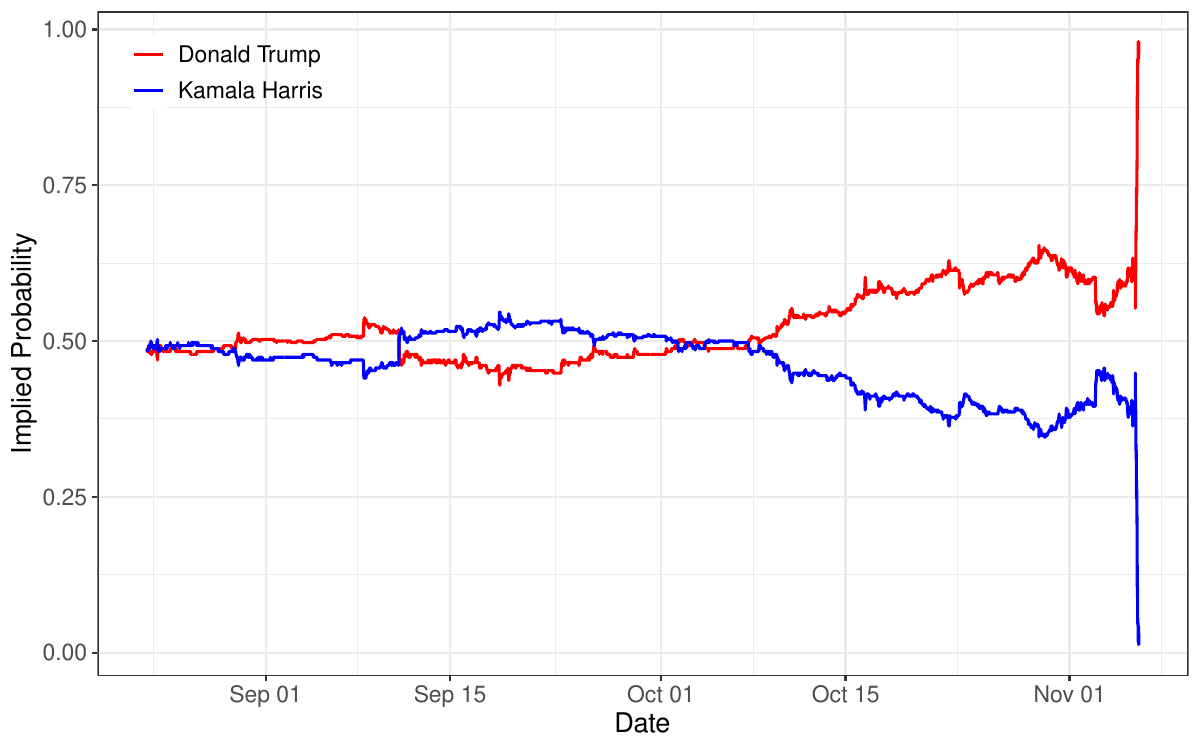}
    \captionsetup{justification=centering}
    \caption{Implied probabilities of Trump and Harris in the 2024 U.S.\ presidential election betting market.}
    \label{fig:2024-data}
\end{figure}

\Cref{fig:2024-data} shows the implied probabilities of the two major-party candidates over the observation period. The horizontal axis represents the observation date, while the vertical axis gives the implied probability. The red line represents Trump and the blue line represents Harris. The two series exhibit substantial co-movement over the sample, with periods of both convergence and divergence. This pattern provides the empirical motivation for modelling their relationship jointly.

During this period, the combined implied probability of the two nominees remained below one but exhibited an upward trend. The value below one reflects the probability mass that continued to be assigned to candidates or outcomes outside the two-nominee subset. Its upward movement is consistent with declining uncertainty, the gradual reallocation of probability away from non-nominees, risk-hedging behaviour \citep{aiba2018foreign}, speculative tail-risk pricing \citep{kelly2014tail}, and increasing recognition of the Republican-Democratic contest as the dominant binary outcome. Accordingly, the empirical analysis is restricted to the back odds of the Republican and Democratic nominees. 

\subsection{Literature review}
The Ornstein-Uhlenbeck (OU) process was originally introduced to model mean reversion in time series data \citep{uhlenbeck1930theory}. Since then, it has been widely applied in financial markets, particularly in interest rate modelling, volatility dynamics, and portfolio management. First, interest rates often display mean-reverting patterns around long-term equilibrium levels. OU processes are employed to model these interest rate movements, facilitating the pricing of interest rate derivatives and the management of interest rate risk \citep{vasicek1977equilibrium}. Second, in stochastic volatility models, OU processes represent the mean-reverting nature of market volatility, which is critical for option pricing and other derivative valuation \citep{fouque2000mean}. Third, in portfolio optimization, asset returns can be modelled with OU processes to incorporate mean-reverting behaviour into allocation and risk management strategies \citep{silva2019modeling}.

Building on the general description above, the pairs-trading literature has developed both empirical and model-based approaches to identifying and trading mean-reverting spreads. The modern quantitative formulation of the strategy is commonly traced to work conducted at Morgan Stanley during the 1980s, providing an influential empirical evaluation \citep{gatev2006pairs}. In model-based approaches, an Ornstein-Uhlenbeck process is often used to represent the spread between the two assets, rather than the individual asset prices. Deviations of the spread from its estimated equilibrium can then be used to define trading entry and exit signals \citep{elliott2005pairs, holy2025estimation}. The Royal Dutch-Shell pair provides a well-known illustration of this statistical-arbitrage framework \citep{reverre2001complete}. At very high sampling frequencies, however, transactions may occur at irregular intervals, making equally spaced discrete-time approximations inappropriate and complicating the estimation of mean-reversion and volatility parameters \citep{engle2000econometrics}. Empirical analysis using Brazilian equity data has also shown that the performance of pairs-trading strategies may depend on the sampling frequency and characteristics of the market under study \citep{perlin2009evaluation}.

In peer-to-peer gambling markets, betting odds generate high-frequency price series that share several empirical features with financial-market data \citep{angelini2022informational}. Existing research has primarily examined market efficiency, including whether historical odds and combined forecasts contain information about subsequent price movements or returns \citep{vlastakis2009efficient}. Financial concepts such as constant relative risk aversion and statistical arbitrage have therefore also been adapted to betting-market settings \citep{hanke2019numeraire, kets2014betting}. Two conceptually distinct forms of arbitrage should nevertheless be distinguished. Statistical arbitrage seeks positive expected returns from predictable or mean-reverting movements in betting prices, but remains risky because the anticipated convergence may not occur. Cross-price arbitrage, by contrast, exploits contemporaneous inconsistencies in the prices across mutually exclusive outcomes or trading venues. Under idealized simultaneous execution, sufficient liquidity, and the absence of transaction costs and execution risk, positions covering all possible outcomes may lock in a payoff independently of the realized outcome \citep{forrest201710}. The present study concerns statistical arbitrage; it uses a time-series model to identify potentially favourable future movements in odds and does not claim to construct a risk-free all-outcomes position.

\section{Methodology}

In this section, we present the methodology used in the data analysis. We first define the gambling market variables and discuss their equilibrium levels. We then introduce the statistical model, estimation procedure, and trading strategy. All analyses were conducted in R version 4.5.1. The statistical modeling, simulation, bootstrap procedures, and graphical analyses were implemented in R.

\subsection{Statistical preliminaries}
\subsubsection{Mean reversion}\label{fair-market benchmark}

In a gambling market, a gambler stakes money on a particular outcome. For example, suppose that a gambler stakes one pound on outcome $i$ at decimal odds of $3.0$. If the selected outcome occurs, the gambler receives three pounds in total, including the original one-pound stake. The resulting profit is therefore two pounds. Otherwise, the one-pound stake is lost. Let $O_{i,t}$ denote the decimal odds for outcome $i$ at time $t$. If a gambler stakes an amount $r$, the profit $F$ is
\begin{align*}
F=
\begin{cases}
    (O_{i,t}-1)r, & \text{if outcome $i$ occurs},\\
    -r, & \text{otherwise}.
\end{cases}
\end{align*}

We define the inverse odds by
\begin{align*}
    X_{i,t}=\frac{1}{O_{i,t}}.
\end{align*}
Inverse odds are often interpreted as probabilities implied by the market \citep{williams1999information}. They need not coincide with the true probabilities of the outcomes.

For a market with $n$ outcomes, we define the aggregate inverse-odds process by
\begin{align*}
    Y_t
    =
    \sum_{i=1}^{n}X_{i,t}
    =
    \sum_{i=1}^{n}\frac{1}{O_{i,t}}.
\end{align*}

We then turn to the short-run dynamics of $Y_t$. Betting-exchange prices can reverse after temporary movements, which is consistent with mean-reverting behavior. For example, anti-persistence has been observed in Betfair odds returns, together with mean-reverting behavior in inverse odds \citep{hardiman2010long}. Betting-exchange prices can also exhibit temporary mispricing \citep{angelini2022informational}. Similar short-run reversals are visible in the political betting in
\Cref{fig:2020-data} and \Cref{fig:2024-data}. Changes in individual inverse odds affect $Y_t$ directly. Their short-run corrective behavior therefore motivates a mean-reverting model for $Y_t$.

We now consider the equilibrium level of $Y_t$ under a fair market. Let $p_{i,t}$ denote the true probability that outcome $i$ occurs. The expected profit from a stake $r$ at odds $O_{i,t}$ is
\begin{align*}
    \mathbb{E}(F)
    =
    r(O_{i,t}p_{i,t}-1).
\end{align*}

Let $O^*_{i,t}$ denote the fair decimal odds for outcome $i$, defined as the odds for which the expected profit is zero. Setting $\mathbb{E}(F)=0$ gives
\begin{align*}
    O^*_{i,t}=\frac{1}{p_{i,t}}.
\end{align*}

For a complete set of mutually exclusive and collectively exhaustive outcomes, the true probabilities $p_{i,t}$ sum to one. Hence, one is the fair-market benchmark for the full-market inverse-odds sum. The observed process $Y_t$ may fluctuate around this level and need not equal one at every observation. If the observed inverse odds fluctuate around their fair values without systematic bias, then
\begin{align*}
    \mathbb{E}(Y_t)=1.
\end{align*}

In a betting exchange, the equilibrium level may differ from the fair-market benchmark. Unlike a bookmaker, a P2P betting exchange does not set a single set of odds with an explicit margin. Instead, the odds are formed through orders submitted by market participants and matched in the order book \citep{franck2013inter}. The resulting prices depend on market microstructure, including bid-ask spreads and available liquidity \citep{marginson2013exchange,flepp2017liquidity}. Consequently, the inverse of the best available odds across all outcomes need not sum to one.

We measure this departure using the over-round \citep{koning2023betting}, defined for a complete set of market outcomes as
\begin{align*}
    \Omega_t
    =
    Y_t-1
    =
    \sum_{i=1}^{n}\frac{1}{O_{i,t}}-1.
\end{align*}

A positive over-round value indicates that the inverse-odds sum lies above the fair-market benchmark. In a P2P betting exchange, this excess arises from the market quotes rather than from a margin imposed directly by the exchange. Market participants who provide liquidity are not required to quote at fair odds. They may quote prices that leave them a positive expected margin for supplying liquidity. Competition and arbitrage tend to reduce this margin, but they do not require it to disappear completely \citep{marginson2013exchange}.

Bid-ask spreads, limited liquidity, and differences in participants' valuations can therefore keep the inverse-odds sum away from the fair-market benchmark. A sum below one may offer an arbitrage opportunity, subject to liquidity, commission, and execution constraints. A sum above one gives a positive over-round \citep{koning2023betting}. Positive over-rounds have been documented empirically on Betfair \citep{franck2013inter}. The exchange platforms earn revenue separately through commission rather than through the over-round \citep{liu2026comparison}.

We therefore allow the equilibrium over-round, $\bar{\Omega}(t)$, to be positive. The corresponding equilibrium level of the full-market inverse-odds sum is
\begin{align*}
    1+\bar{\Omega}(t).
\end{align*}

However, the empirical analysis in this paper focuses on the two major-party nominees, who represent the dominant outcomes in the presidential election market. Other listed outcomes also contribute to the full-market inverse-odds sum, although their contribution is generally smaller. The inverse-odds sum for the two major candidates is therefore not required to exceed one, even when the full-market equilibrium is above one. Its equilibrium level depends on both the market over-round and the contribution of the remaining outcomes. Since these quantities may change over time, we allow the equilibrium level of the two-candidate process to vary with time. Thus, we model short-run deviations from this level as mean reverting.

\subsubsection{Ornstein-Uhlenbeck process with additive white noise}\label{sec:ou}
The equilibrium properties of $Y_t$ are summarized in the preceding subsection, provided in \cref{fair-market benchmark}. The empirical evidence discussed above provides separate motivation for modelling short-run deviations from the equilibrium as mean reverting. The Ornstein-Uhlenbeck process provides a continuous-time representation of such behavior. It has also been used to model betting-odds dynamics. For example, a time-dependent Ornstein-Uhlenbeck process has been applied to horse-racing betting odds \citep{sugawara2025ornstein}. We therefore assume that the latent state $P_t$ follows an Ornstein-Uhlenbeck process and that $Y_t$ is observed with additive white noise:
\begin{align}
    \d P_t 
    & = \theta (\mu-P_t)\d t + \sigma \d W_t, \label{eq:ou1} \\
    Y_t & = P_t + E_t, \nonumber
\end{align}
where $W_t$ is a Wiener process, $\mu$ is the equilibrium level to which the process reverts, $\theta$ is the rate of mean reversion, and $\sigma$ is the diffusion coefficient controlling the magnitude of the stochastic fluctuations. The term $E_t\stackrel{\text{i.i.d}}{\sim} \mathcal{N}(0,\omega^2)$ represents additive observation noise, where $\omega^2$ denotes its variance. In this setting, the observation noise captures short-term market microstructure effects not represented by the latent OU process \citep{holy2025estimation}. 

The Ornstein-Uhlenbeck (OU) process with additive white noise is well-suited for mean-reverting data, as it combines a stabilizing drift that pulls values back toward equilibrium with a noise term that captures random fluctuations. This balance makes it an effective model for systems such as implied probabilities, which exhibit both short-term volatility and long-term reversion \citep{gray2008macrofinancial}. This model was previously considered by Hol\'{y} and Tomanov\'{a} \citet{holy2025estimation} in the context of pairs trading. We note that whereas in our case the \emph{sum} of implied probabilities has an equilibrium, in pairs trading it is the \emph{difference} in the instruments traded (e.g., stocks) that has an equilibrium. 

\subsubsection{Trending OU process with noise}\label{trending}
A constant over-round can shift the equilibrium level of the full-market implied probability sum above one, but it does not by itself create a time trend. However, in political betting markets, the equilibrium level may evolve over time as uncertainty declines and probability mass is gradually reallocated from outcomes outside the candidate subset to the leading candidates. In such cases, a time-varying equilibrium provides a more appropriate modelling framework, thereby motivating the use of a time-varying-mean Ornstein-Uhlenbeck process with additive white noise to capture both mean reversion and evolving market dynamics.

In U.S.\ presidential elections, the implied probabilities of all leading nominees are influenced not only by short-term fluctuations and arbitrage, which generate mean reversion, but also by gradual structural shifts in expectations. This is because the equilibrium probability is influenced by persistent factors beyond short-term market fluctuations, including the following.
\begin{enumerate}
    \item Ongoing information flow: Continuous updates from polls, campaign events, debates, and media coverage gradually shift market expectations \citep{zaller1989bringing}.
    \item Structural trends in voter sentiment: Momentum effects and slowly accumulating public opinion induce persistent drift in the perceived equilibrium probabilities \citep{stimson2018public}.
\end{enumerate}

So far we have considered the total probability $Y_t = \sum_{i=1}^n 1/O_{i,t}$, defined by summing the inverse (decimal) odds across all candidates. Now we consider a subset $\mathcal{I} =\{1,2\}$ of candidates, in particular the Republican and Democratic nominees, and define the \emph{combined probability}
\begin{align*}
     Z_t 
    := \sum_{i \in \mathcal{I}} \frac{1}{O_{i,t}}.
\end{align*}
For $\mathcal{I}$ to be the leading candidates $(n=2)$, $Y_t = Z_t+Z^*_t$ where $Z^*_t$ is the sum of implied probabilities of all other candidates. 

Under the deterministic specification for the contribution of the remaining outcomes stated in \ref{trending OU with noise}, a time-varying-mean OU model for the full-market process induces a corresponding time-varying-mean OU representation for the two-candidate process. First, restricting the analysis to the two major-party nominees can contribute to the observed trend.The equilibrium level of the two-candidate process may therefore vary over time, as detailed in \ref{trending OU with noise}. Second, the continuing arrival of political information and gradual changes in public opinion may cause the equilibrium level of election-market probabilities to vary over time \citep{zaller1989bringing, stimson2018public}. This possibility motivates the use of a time-varying mean-reversion level, rather than implying that the equilibrium must drift monotonically in a particular direction.

Based on the model for the total probability $\{Y_t\}$, defined in \cref{eq:ou1} with the assumptions stated in \ref{trending OU with noise}, we derive that $\{Z_t\}$ follows a \emph{trending OU process with additive white noise}.  This is defined by  
\begin{align}
    \d Q_t 
    & = \theta \{\mu(t)-Q_t\} \, \d t + \sigma \d W_t, \label{eq:ou2} \\
    Z_t & = Q_t + E_t, \nonumber
\end{align}
where $\mu(\cdot)$ is a not necessarily constant function of time $t$, and the other terms are defined as in \cref{eq:ou1}. Prediction-market prices aggregate dispersed information and adjust as new information becomes available \citep{wolfers2004prediction, pennock2002modelling}. Consequently, the equilibrium level of an aggregate implied probability measure need not remain constant over the life of the market. We therefore posit that short-run deviations of $Z_t$ from a deterministic time-varying mean-reversion level are mean reverting, and model $Z_t$ using a trending Ornstein-Uhlenbeck process with additive white noise. The trending Ornstein-Uhlenbeck process allows the mean-reversion level $\mu(t)$ to vary over time, while preserving local mean-reverting behaviour \citep{thamrongrat2023application}.

\subsubsection{Parameter estimation}\label{parameter-estimation}


The observed series can include flat periods in which the values, and therefore $Z_t$, remain unchanged across several consecutive observations. We define a flat period as two or more consecutive scheduled observations with the same value. For likelihood-based parameter estimation, we retain one observation from each flat period. We also preserve the number of original observation intervals between successive retained observations.

Consecutive repeated values are represented by a single retained observation. Let $\tau_j$ denote the original observation index of the $j^{\text{th}}$ retained value, where $j=1,\ldots,J$, and let $J$ denote the total number of retained observations. The elapsed interval between two successive retained observations is
\begin{align*}
    q_j=\tau_j-\tau_{j-1}.
\end{align*}

Thus, $q_j$ gives the number of original observation intervals between the retained observations. This interval is used directly in the continuous-time OU state transition.

As described in \ref{One-step Conditional Moments}, the likelihood is constructed from the conditional distributions
\begin{align*}
    \mathbb{P}\left(
         Z_j
        \mid
         Z_1,\ldots, Z_{j-1};
        q_j
    \right),
    \quad
    j=2,\ldots,J.
\end{align*}
Each retained observation is evaluated conditional on the previously retained observations, with the elapsed observation interval incorporated through $q_j$.

We assume that genuinely missing observations satisfy the Missing Completely at Random (MCAR) condition, such that the probability that a scheduled observation is unavailable is independent of both the observed and unobserved values of the process \citep{heitjan1996distinguishing}. Under this assumption, the missing observations are treated as ignorable for likelihood-based inference conditional on the observed sampling times. This formulation accounts for the unequal elapsed times between successive quote changes without treating flat periods as missing observations. The estimation then proceeds sequentially in two steps.

In the first step, we model the conditional expectation of the process using the moment condition to get the marginal mean of $Z_t$, and we define $M\left(t\right)=\mathbb{E}\left(Z_t\right)$. To capture potential nonlinearity and smooth temporal variation, $M(t)$ is modelled non-parametrically using a smoother (penalized thin plate regression spline). 
This specification allows the model to flexibly accommodate gradual structural changes in the data without imposing restrictive parametric forms. For a standard Ornstein-Uhlenbeck process with additive white noise and a constant mean-reversion level, the unconditional expectation converges asymptotically to the stationary long-run mean $\mu$. By contrast, in a trending Ornstein-Uhlenbeck process the mean-reversion level is time-varying, and the state variable follows this evolving benchmark through an exponential adjustment mechanism. As a result, the expected value corresponds to a smoothed transformation of $\mu(t)$ rather than $\mu(t)$ itself. We know $\mathbb{E}(Q_t) = M(t)$ since $E_t$ is mean zero. Then we take expectations of \cref{eq:ou2}, and we get
\begin{align*}
    \d \mathbb{E}(Q_t)=
    \theta\{\mu(t)-\mathbb{E}(Q_t)\}\,\mathrm{d}t+0,
\end{align*}
because $\d W_t$ is white noise. Since $\mathbb{E}[\d Q_t]=\d M(t)$, it follows that $M(t)$ satisfies the ordinary differential equation
\begin{align}
    M^{\prime}(t)=\theta\{\mu(t)-M(t)\}. \label{eq:ODE}
\end{align}
The first-step estimation of $M(t)$ provides a data-driven characterization of the long-run trend in $\mathbb{E}[Z_t]$, which is then incorporated into the second-step likelihood estimation of the stochastic dynamics of $\{Z_t\}$.

In the second step, the estimated smooth mean function $\hat M(t)$ obtained from the first-stage smooth is treated as fixed and substituted into the likelihood function of the stochastic process $\{Z_t\}$. According to \cref{eq:ODE}, we build 
\begin{align}
   \hat \mu(t,\theta)=\hat M(t)+\frac{1}{\theta}\hat M^{\prime}(t). \label{eq:newmu}
\end{align} 
Parameter estimation then proceeds by maximizing the log-likelihood with respect to the structural parameter $\theta$ and the variance components $\sigma^2$ and $\omega^2$, while treating the estimated mean function $\hat M(\cdot)$ as fixed,
\begin{align}
        (\hat{\theta},\hat{\sigma}^2,\hat{\omega}^2)'
    = \arg\underset{\theta,\sigma^2,\omega^2} {\max}L(\theta,\sigma^2,\omega^2;\hat M(\cdot)) \quad \text{s.t.} \quad \sigma^2 > 0, \omega^2 > 0. \label{eq7:MLE}
\end{align}
After that, we plug $\hat \theta$ into \cref{eq:newmu} and get the resulting estimator of $\mu(t)$:
\begin{align}
    \hat \mu(t)= \hat M(t)+\frac{1}{\hat \theta}\hat M^{\prime}(t). \label{eq:m(t)}
\end{align}

The two-step procedure alternates between smooth mean estimation and likelihood-based estimation of the stochastic parameters until the convergence criterion is met. We use the limited-memory Broyden-Fletcher-Goldfarb-Shanno (L-BFGS-B) algorithm for numerical optimization. L-BFGS-B is a quasi-Newton method that approximates the Hessian matrix of the objective function and allows simple parameter bounds to be imposed \citep{saputro2017limited}. A Kalman filter is used to evaluate the likelihood while accounting for variation in the sampling intervals between successive retained observations \citep{kalman1960new}. This interval determines the corresponding state-transition coefficient and process-noise variance in \ref{One-step Conditional Moments}.

\subsection{Betting strategy}\label{betting strategy}

We treat decimal odds as tradable price quotes and formulate a simplified buy-sell strategy analogous to trading a financial asset. This representation is an abstraction of betting-exchange trading and does not reproduce the complete cash flows of matched back and lay bets, including lay liability, bid-ask spreads, commission, queue priority, partial execution or settlement at event resolution. The resulting return measures only the proportional change in the selected decimal-odds quote between entry and exit. It does not represent the profit that would necessarily be obtained from an executed trade on a betting exchange, because it does not account for bid-ask spreads, commission, queue priority, partial execution, or other trading frictions \citep{liu2026comparison}.

Suppose that a statistical signal is generated at time $t$ and that the pairwise comparison model selects candidate $i$. The observed decimal odds $O_{i,t}$ are taken as the entry value for evaluating the subsequent odds-price movement.

The odds series often contains flat periods. Therefore, the quote observed
at $t+1$ may be identical to the entry quote. Closing every position mechanically at $t+1$ would produce many zero price changes that simply reflect unchanged quotes rather than a movement following the signal. We therefore define the exit lag as
\begin{align*}
    k=
    \inf\left\{
        h\in\mathbb{N}_{+}:
        O_{i,t+h}\neq O_{i,t}
    \right\}.
\end{align*}
The exit time $t+k$ is the first subsequent observation at which the selected candidate's odds differ from the entry value. This rule allows the return calculation to capture the first observed price movement after the signal.

The value of $k$ is not known when the position is opened. It is determined sequentially as new observations become available. The procedure therefore does not use future information when the entry decision is made. A positive odds-price return occurs when $O_{i,t+k}>O_{i,t}$. The framework allows only one open position at a time. Any new signals that occur before the current position is closed are ignored. A new position can be opened from the first forecast origin after the previous position has been closed. If the selected odds do not change again before the end of the sample, the position is classified as incomplete and is excluded from the realised-return calculation.

We propose the principle that bets should be placed only when they have a high probability of being profitable. 
We take the perspective that realistically modelling the odds $O_{i,t}$ is untenable so proceed by employing the distribution of the combined probability $Z_t$ and some regression modelling described later. As we will see, this naturally decomposes the problem into two components: determining when to place a bet, then determining which bet to place.  For simplicity we do not consider transaction fees that are incurred when trading between buy and sell odds. Gamblers do not simultaneously place opposing bets; for instance, a bet placed on the Republican nominee precludes placing a bet on the Democratic nominee at the same time. To provide trading recommendations at each time point, we adopt an online learning framework that dynamically determines the signal timing for bet placement.

\subsubsection{When to place a bet}

In this section, we determine when a statistical signal is generated. The one-step predictive distribution is defined over the next scheduled observation interval. Thus, $Z_{t+1}$ denotes the combined implied probability at the next observation after time $t$. A decrease in the combined implied probability, $Z_{t+1}<Z_t$, implies that the implied probability of at least one candidate decreases. The corresponding decimal odds must therefore increase for at least one $i\in\mathcal{I}$. We use this relationship to define the trading signal. A signal is generated at time $t$ when a decrease in $Z_{t+1}$ is predicted with sufficiently high probability.

Let $\hat F_{t+1\mid t}$ denote the estimated predictive distribution of $Z_{t+1}$ conditional on the observations $Z_1,\ldots,Z_t$. The corresponding upper $(1-\alpha)$ predictive quantile is
\begin{align*}
    \hat U_{t+1\mid t}^{1-\alpha}=
    \hat F_{t+1\mid t}^{-1}(1-\alpha).
\end{align*}
We define the statistical signal as
\begin{align*}
    S_t=\mathbf{1}
    \left\{
        Z_t>
        \hat U_{t+1\mid t}^{1-\alpha}
    \right\}.
\end{align*}

Under a continuous predictive distribution, with $\alpha\in(0,1)$, $S_t=1$ implies
\begin{align*}
    \mathbb{P}
    \left(
        Z_{t+1}\le Z_t
        \mid
        Z_1,\ldots,Z_t
    \right)
    = 1-\alpha.
\end{align*}
The probability is conditional on the observations available up to time
$t$. Since $\{Z_t\}$ exhibits mean-reverting behaviour, a value of $Z_t$
above its time-varying equilibrium level generally increases the probability of a subsequent downward movement. The corresponding predictive quantile therefore defines a one-sided upper prediction bound
with nominal coverage $1-\alpha$.

In time series analysis, bootstrapping is a widely used technique for constructing prediction intervals. It estimates the distribution of an estimator by resampling, often with replacement, from the observed data or from a model fitted to the data \citep{thombs1990bootstrap}. Given the observed series $\{Z_1,\ldots,Z_n\}$, let $n$ denote the total number of observations and let $l<n$ denote the initial training-sample size. At each forecast origin $t=l,\ldots,n-1$, we apply the following parametric bootstrap procedure \citep[see, e.g.,][p.~169]{efron2021computer}:

\begin{enumerate}

\item For a given forecast origin $t$, use the observations $Z_1,\ldots,Z_t$ as the available data history.

\item Fit the smooth expectation function $\hat M^{(0)}(\cdot)$ to $Z_1,\ldots,Z_t$. Estimate $\hat\theta^{(0)}$, $\hat\sigma^{2(0)}$, and $\hat\omega^{2(0)}$ by maximum likelihood.

\item Apply the Kalman filter recursions in
\cref{innovation covariance} to $Z_1,\ldots,Z_t$. This gives the
one-step-ahead state prediction $\hat Q_{t+1\mid t}^{(0)}$ and the
prediction-error variance $P_{t+1\mid t}^{(0)}$. The predictive distribution of the next observation is
\begin{align*}
Z_{t+1}
\mid
Z_1,\ldots,Z_t
\sim
\mathcal{N}\left(
\hat Q_{t+1\mid t}^{(0)},
P_{t+1\mid t}^{(0)}
+
\hat\omega^{2(0)}
\right).
\end{align*}

\item For $b=1,\ldots,B$, perform the following steps:
\begin{enumerate}

\item Generate a bootstrap series recursively. Set $Z_1^{(b)}=Z_1$. For $s=1,\ldots,t-1$, update the filtered state using $Z_1^{(b)},\ldots,Z_s^{(b)}$. Propagate the state to time $s+1$ and draw
\begin{align*}
Z_{s+1}^{(b)}
\sim
\mathcal{N}\left(
\hat Q_{s+1\mid s}^{(0,b)},
P_{s+1\mid s}^{(0,b)}
+
\hat\omega^{2(0)}
\right).
\end{align*}

\item Refit the model to $Z_1^{(b)},\ldots,Z_t^{(b)}$. This gives
$\hat M^{(b)}(\cdot)$, $\hat\theta^{(b)}$,
$\hat\sigma^{2(b)}$, and $\hat\omega^{2(b)}$.

\item Apply the Kalman filter to the original observed history $Z_1,\ldots,Z_t$ using the bootstrap parameter estimates. This gives $\hat Q_{t+1\mid t}^{(b)}$ and
$P_{t+1\mid t}^{(b)}$.

\item Draw the bootstrap prediction
\begin{align*}
Z_{t+1}^{*(b)}
\sim
\mathcal{N}\left(
\hat Q_{t+1\mid t}^{(b)},
P_{t+1\mid t}^{(b)}
+
\hat\omega^{2(b)}
\right).
\end{align*}

\end{enumerate}
\end{enumerate}

Then we compute the upper one-sided prediction bound as 
\begin{align}
    \hat{U}_{t+1}^{1-\alpha}=\min\left\{z_{t+1} : \hat{F}_{t\mid t-1}(z_t)\ge 1-\alpha \right\}, \label{eq9:prediction bound}
\end{align}
where $\hat{F}_{t\mid t-1}(z_t)=\frac{1}{B}\sum_{b=1}^{B}\mathbf{1}\left(\hat{Z}_{t+1}^{*(b)}\le z_{t+1}\right),$ so that the estimated prediction interval is $(-\infty, \hat{U}_{t+1}^{1-\alpha})$. The signal to place a bet at time $t$ is that the current value $Z_{t}$ exceeds the upper predictive bound for $Z_{t+1}$, that is, when $Z_{t} > \hat{U}_{t+1}^{1-\alpha}$. This expression indicates that the combined probability will decrease from $Z_t$ with probability at least $1-\alpha$.

To assess the calibration of these prediction intervals, we calculate their empirical coverage over the set of out-of-sample forecast origins $\mathcal{T}$:
\begin{align*}
    \hat{C}_{1-\alpha}=\frac{1}{|\mathcal{T}|}\sum_{t\in\mathcal{T}}\mathbf{1}
    \left\{
        Z_{t+1}
        \leq
        \hat{U}_{t+1}^{1-\alpha}
    \right\}.
\end{align*}
A well-calibrated upper prediction interval should have empirical coverage close to its nominal level $1-\alpha$.

\subsubsection{Which bet to place}

Once a statistical signal has been identified, the next step is to determine which candidate-specific odds to select. Let $Y_{i,t}=\frac{1}{O_{i,t}}, \quad i\in\{1,2\}$ denote the implied probability of candidate $i$ at signal time $t$, where $i=1$ denotes the Republican nominee and $i=2$ denotes the Democratic nominee. Let $\S$ denote the set of signal times. At time $t$, the most recently observed change in the implied probability
of candidate $i$ is defined as
\begin{align*}
    \Delta Y_{i,t}
    =
    Y_{i,t}-Y_{i,t-1}.
\end{align*}
The candidate-specific feature vector is therefore
\begin{align*}
    X_{i,t}
    =
    \begin{pmatrix}
        Y_{i,t}\\
        \Delta Y_{i,t}
    \end{pmatrix}
    \in\mathcal X,
    \quad i\in \{1,2\},
\end{align*}
where $\mathcal X\subseteq\mathbb R^2$ denotes the individual feature space.

For model training, define the subsequent change in the implied probability by
\begin{align*}
    \Delta^{+} Y_{i,t}
    =
    Y_{i,t+1}-Y_{i,t}.
\end{align*}
A smaller value of $\Delta^{+} Y_{i,t}$ represents a larger decrease in the
implied probability. For example, a change of -0.03 represents a larger decrease than a change of -0.01. Candidate 1 is therefore preferred to candidate 2 when $\Delta^{+} Y_{1,t}<\Delta^{+} Y_{2,t}$.

The binary response is defined as
\begin{align*}
    L_t
    =
    \begin{cases}
        1, & \text{if } \Delta^{+} Y_{1,t}<\Delta^{+} Y_{2,t},\\
        0, & \text{otherwise}.
    \end{cases}
\end{align*}

Thus, $L_t=1$ indicates that candidate 1 experiences the larger subsequent decrease in implied probability, whereas $L_t=0$ indicates that candidate 2 experiences the larger subsequent decrease. Observations for which the two subsequent changes are equal are treated as ties and are excluded from the pairwise model fitting \citep{hunter2004mm}.

Because the labels assigned to the two candidates have no intrinsic ordering, the comparison rule should reverse when the order of the candidates is exchanged. We therefore define the pairwise feature difference by
\begin{align*}
    D_t
    =
    X_{1,t}-X_{2,t},
\end{align*}
and the corresponding comparison score by
\begin{align*}
    d(X_{1,t},X_{2,t})
    =
    \boldsymbol{\beta}^{\top}
    \left(
        X_{1,t}-X_{2,t}
    \right),
\end{align*}
where $\boldsymbol{\beta}\in\mathbb R^2$ is estimated from the training data. This construction is antisymmetric because
\begin{align*}
    d(X_{2,t},X_{1,t})
    &=
    \boldsymbol{\beta}^{\top}
    \left(
        X_{2,t}-X_{1,t}
    \right)\\
    &=
    -\boldsymbol{\beta}^{\top}
    \left(
        X_{1,t}-X_{2,t}
    \right)\\
    &=
    -d(X_{1,t},X_{2,t}).
\end{align*}
Thus, exchanging the order of the two candidates reverses the direction of the comparison while preserving its magnitude.

Applying a logistic link to the antisymmetric comparison score gives a Bradley-Terry-type pairwise comparison model
\citep{hunter2004mm}. The conditional probability that candidate 1 is preferred to candidate 2 is
\begin{align}
    \pi_t
    &:=
    \mathbb P
    \left(
        L_t=1
        \mid
        X_{1,t},X_{2,t}
    \right)\\
    &=
    \frac{
        1
    }{
        1+
        \exp\left[
            -\boldsymbol{\beta}^{\top}
            \left(
                X_{1,t}-X_{2,t}
            \right)
        \right]
    }.
    \label{eq13: Bradley-Terry Model}
\end{align}
It follows that reversing the candidate order gives
\begin{align*}
    \mathbb P
    \left(
        L_t=0
        \mid
        X_{2,t},X_{1,t}
    \right)
    =
    1-\pi_t.
\end{align*}

No intercept is included in the pairwise logistic model. If a non-zero intercept $\beta_0$ were included, the comparison score would become
\begin{align*}
    d(X_{1,t},X_{2,t})
    =
    \beta_0
    +
    \boldsymbol{\beta}^{\top}
    \left(
        X_{1,t}-X_{2,t}
    \right),
\end{align*}
whereas reversing the candidate order would give
\begin{align*}
    d(X_{2,t},X_{1,t})
    =
    \beta_0
    -
    \boldsymbol{\beta}^{\top}
    \left(
        X_{1,t}-X_{2,t}
    \right).
\end{align*}
This is not generally equal to $-d(X_{1,t},X_{2,t})$. The intercept is therefore fixed at zero to preserve antisymmetry \citep{hunter2004mm}.

The candidate-selection rule is defined by
\begin{align}
    \delta(X_{1,t},X_{2,t})
    =
    \begin{cases}
        1,
        & \pi_t\ge0.5,\\ 
        0,
        & \pi_t<0.5,
    \end{cases}
    \label{eq:candidate-selection}
\end{align}
where $\delta=1$ denotes selection of the Republican nominee's odds quote, and $\delta=0$ denotes selection of the Democratic nominee's odds quote. Equivalently, candidate 1 is selected when $d(X_{1,t},X_{2,t})\ge0$, whereas candidate 2 is selected when $d(X_{1,t},X_{2,t})<0$.

\section{Case study: application to U.S.\ presidential elections}
\subsection{Training data: 2020 U.S.\ election}
The 2020 data contain flat periods and a small number of genuinely missing scheduled observations. We handle these observations using the procedure described in \cref{parameter-estimation}, with the elapsed number of sampling intervals between successive retained observations preserved in the continuous-time state transition.

We use the first 10,000 observations as an initial training set. This corresponds to approximately the first half of the sample and provides a sufficiently long series for initial parameter estimation while retaining a substantial number of observations for out-of-sample evaluation. The resulting parameter estimates are used to estimate the starting parameters ($\hat\mu(\cdot),\hat\theta,\hat\sigma^2,\hat\omega^2$) for the remaining data. An online learning framework is then applied, updating predictions sequentially at each time point. For each step, we employ bootstrapping to construct an upper one-sided prediction interval and assess whether the subsequent observation falls outside this interval.

\subsubsection{When to place a bet}

After the simulation study in \ref{Data Simulation}, we applied the model to the 2020 U.S.\ presidential election as a case study. In this election, the Republican candidate was Trump, and the Democratic candidate was Biden. The analysis period spanned from $23$ August, when both candidates were officially nominated, to 4 November 2020, the day following Election Day. Odds were recorded at five-minute intervals, yielding a total of 21,181 observations. The regular five-minute grid was retained for defining forecast origins and evaluating one-step-ahead prediction coverage. For parameter estimation, consecutive repeated values within each flat period were represented by a single retained observation. The elapsed number of observation intervals between successive retained observations was incorporated into the continuous-time OU state transition. When a statistical signal generated an entry, the selected odds quote was monitored until its first subsequently observed change. Repeated observations at the entry quote therefore extended the holding period of the position but were not treated as additional entries while
the position remained open. In addition, an unpredictable event occurred on 2 October 2020,  when Trump tested positive for COVID-19, leading to a sharp decline in the implied probability. Following his recovery on 6 October, the probability increased again. Aside from this short episode, the overall sum of implied probabilities displayed an upward trend. Using the trending OU process with additive white noise, we identified 107 signal points $\S$. At each signal point $t$, we took a unit-notional synthetic position in the selected decimal-odds quote and closed the position according to the pre-specified exit rule in \cref{betting strategy}. 

If we exclude the period when Trump contracted COVID-19, the resulting plot remains well described by the trending OU process with additive white noise and there are 201 trading signal points $\S$ out of 4,135 non-flat time observations (the first point of flat periods), as shown in \Cref{fig:2}. Removing the COVID-19 interval changes the fitted trend and predictive distribution over the remaining sample; the two signal counts are therefore obtained from separately fitted analyses and are not directly nested. The empirical coverage rate is 95.2\%, close to the nominal 95\% level specified in \ref{Model Selection}. The COVID-19 episode motivates two possible strategies for handling extreme events in peer-to-peer gambling markets.
\begin{enumerate}
    \item Event-driven betting. Gamblers may choose to place bets immediately following the occurrence of extreme events. As shown in \Cref{fig:2}, the entire period associated with Trump's COVID-19 diagnosis generated signals to place bets. However, a key challenge lies in determining when such extreme events have concluded.
    \item Risk-avoidance strategy. Alternatively, gamblers may prefer to abstain from betting during extreme events to avoid exposure to heightened uncertainty. From \Cref{fig:2}, we observe that following the specified COVID-19 interval, the observed series again displayed fluctuations around the fitted time-varying mean-reversion level. This post-event sequence is again well captured by the trending OU process with noise, providing a stable foundation for strategic betting.
\end{enumerate}

\begin{figure}[htbp]
    \centering
    \includegraphics[width=5in]{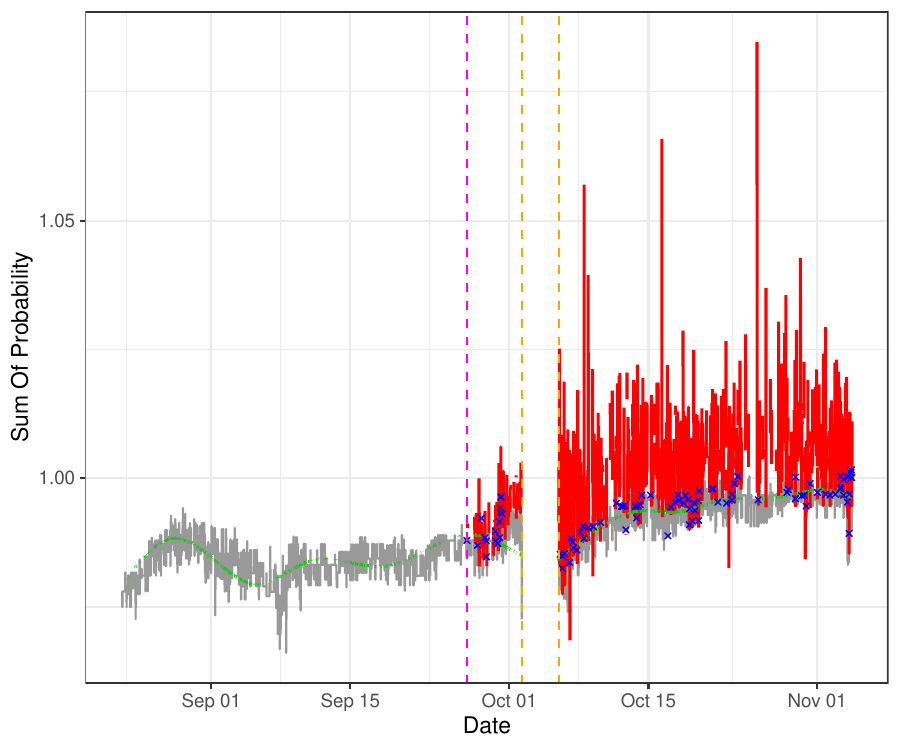}
    \captionsetup{justification=centering}
    \caption{Betting signals and model fit for the 2020 U.S.\ presidential election. The grey line shows the observed combined implied probability. The green line shows the estimated time-varying mean-reversion level $\mu(t)$, and the red line shows the upper one-step-ahead prediction bound. The blue crosses mark the signal points. The purple dashed vertical line marks the beginning of the out-of-sample prediction period, while the two orange dashed vertical lines mark 2 and 6 October, respectively, corresponding to the COVID-19 period.}
    \label{fig:2}
\end{figure}

\subsubsection{Estimator performance}
In the first step, we estimate $\hat{M}(t)$, where the point estimate of the long-run mean is 0.988. In the second step, we employ the L-BFGS optimization method to obtain the maximum likelihood estimators for $(\theta,\sigma^2,\omega^2)$. After that, we substitute $\hat \theta$ into \cref{eq:m(t)} to get $\mu(t)$. Based on \Cref{fig:2}, $\mu(t)$ exhibits an upward trajectory, increasing from approximately $0.975$ to a value close to 1, which is consistent with the trending assumption in \cref{trending}. The results are reported in \Cref{table:2020}, which presents the final point estimates of the parameters. The online-learning parameter estimation procedure required approximately 25,200 seconds of computation time on an Intel Core i7-6700HQ processor (2.60 GHz) with 24 GB of RAM. We employ a parametric bootstrap to estimate standard errors, generating synthetic samples from the fitted trending OU process with additive noise model evaluated at the final parameter estimates in \Cref{table:2020}.

\begin{table}[htbp]
    \caption{Parameter estimates for the 2020 U.S.\ presidential election}
    \label{table:2020}
    \centering
    \begin{tabular}{ccccc}
    \hline
    Parameter & Point estimate & S.E. & 95\% CI lower & 95\% CI upper \\
    \hline
    $\theta$
    & 0.134
    & 0.00220
    & 0.130
    & 0.139 \\
    \hline
    $\sigma^2$
    & $4.83 \times 10^{-6}$
    & $2.64 \times 10^{-7}$
    & $4.34 \times 10^{-6}$
    & $5.38 \times 10^{-6}$ \\
    \hline
    $\omega^2$
    & $2.73 \times 10^{-6}$
    & $2.11 \times 10^{-7}$
    & $2.31 \times 10^{-6}$
    & $3.14 \times 10^{-6}$ \\
    \hline
    \end{tabular}
\end{table}

\subsubsection{Which bet to place}
From \Cref{fig:2}, we observe that after 10,000 time points, the inverse odds of Biden (Democratic Party) exceed those of Trump (Republican Party). This indicates that in the peer-to-peer gambling exchange, market participants increasingly believed that Biden was more likely to win the election. As a result, the price of betting on Biden became lower than that of Trump. In practice, if we excluded extreme events of the COVID-19 period, after applying the Bradley-Terry Model from \cref{eq13: Bradley-Terry Model}, we used it to determine which bet to place. At each signal time $t$, the fitted Bradley-Terry-type model produces the conditional preference probability $\pi_t$ defined in \cref{eq13: Bradley-Terry Model}. Candidate 1 represents the Republican nominee, whereas candidate 2 represents the Democratic nominee. The estimated coefficients obtained from the 2020 U.S.\ presidential election training data are reported in \Cref{table:4}. The candidate-specific odds are selected according to \cref{eq:candidate-selection}. Specifically, the Republican nominee's odds are selected when $\pi_t\geq0.5$, whereas the Democratic nominee's odds are selected when $\pi_t<0.5$. Equivalently, candidate 1 is selected when $d(X_{1,t},X_{2,t})\geq0$, and candidate 2 is selected when $d(X_{1,t},X_{2,t})<0$.

\subsubsection{Model diagnostics}
To evaluate the adequacy of the fitted trending Ornstein-Uhlenbeck (OU) process with additive white noise and the Bradley-Terry model, we conducted the diagnostic checks reported in \ref{Model Diagnostic}. For the trending OU model, we examined the standardized residuals using a Q-Q plot and an autocorrelation function (ACF) plot. The central part of the Q-Q plot is approximately linear, indicating that the Gaussian distribution provides a reasonable description of the central portion of the residual distribution. However, substantial deviations occur in both tails, particularly in the lower tail, indicating heavier-than-normal tails and the presence of several extreme residuals. The Gaussian assumption should therefore be regarded as an approximation that does not fully capture the tail behaviour.

The residual ACF is close to zero at most lags, although several autocorrelations exceed the approximate 95\% confidence bounds. This suggests that the model captures most of the temporal dependence in the data, but that some weak residual serial correlation remains. Inspection of the residuals against time and fitted values did not reveal a clear systematic change in variance.

For the Bradley-Terry model, we assessed calibration using a Hosmer-Lemeshow-type diagnostic, which compares observed outcomes with model-implied probabilities across groups with similar predicted probabilities \citep{hosmer1997comparison}. The resulting $p$-value of 0.29 provides insufficient evidence to reject the null hypothesis of adequate calibration. Thus, we found no strong evidence of systematic miscalibration across the probability groups considered. Overall, the diagnostics indicate that the models provide a broadly adequate description of the central structure of the data, although the heavy-tailed residual distribution and the remaining weak serial dependence should be acknowledged.

\subsection{Test data: 2024 U.S.\ election}
\subsubsection{Data analysis}
We use data from the 2024 U.S.\ presidential election to evaluate the proposed selection models. During 2024 election, the Republican candidate was Trump, and the Democratic candidate was Harris. The first out-of-sample signal was generated after 11:15\,a.m. on 26 August 2024, and the final bet was placed before 6:15\,a.m. on 6 November 2024. From \Cref{figure:1}, the red line is the one-sided prediction interval. A total of 130 points fall outside the upper one-sided prediction interval. The empirical coverage rate is 95.1\%, which is slightly larger than the target 95\%. This indicates a slight over-coverage, suggesting that the constructed intervals are slightly conservative. The out-of-sample period produced 130 statistical signals. None of these signals occurred while an existing position remained open, and all selected odds quotes changed again before the end of the sample. Therefore, all 130 signals resulted in completed trades for performance evaluation and can be interpreted as betting signals $\S$.

\begin{figure}[htbp]
    \centering
    \includegraphics[width=5in]{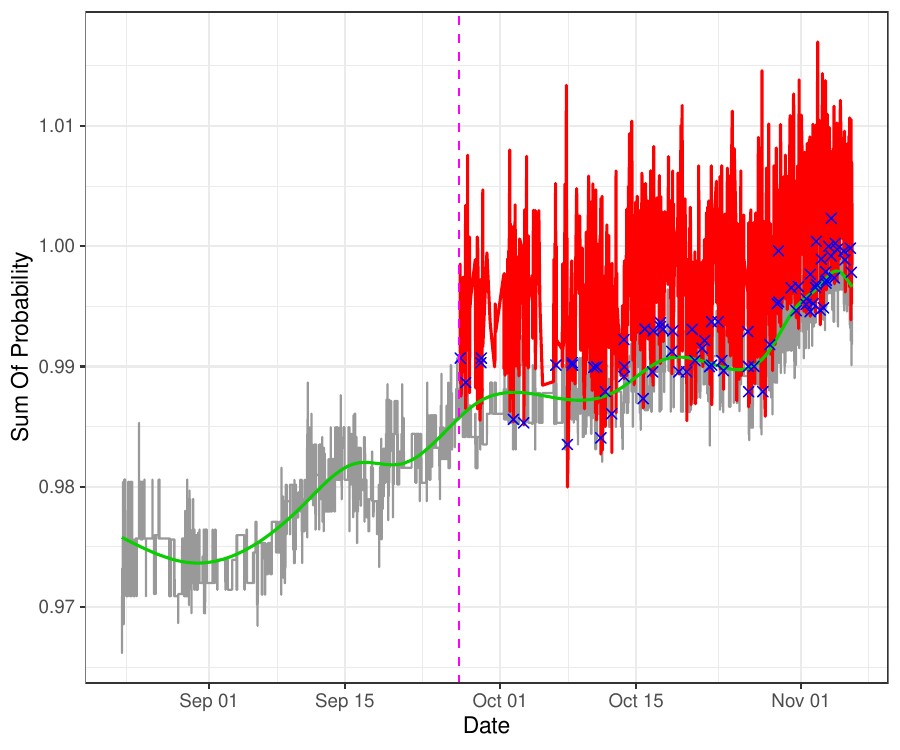}
    \captionsetup{justification=centering}
    \caption{Combined implied probability of the two major-party nominees in the 2024 U.S.\ presidential election. The grey line shows the observed series, the green line shows the estimated time-varying mean-reversion level $\mu(t)$, the red line shows the upper one-step-ahead prediction bound, and the blue crosses identify signal times. The purple dashed vertical line marks the beginning of the out-of-sample forecasting period.}
    \label{figure:1}
\end{figure}


In the first step, we estimated smooth mean function $\hat{M}(t)$, which the fitted mean at the end of the estimation period was 0.985. The second step produced the parameter estimates $(\theta,\sigma^2,\omega^2)$. From \Cref{figure:1}, the time-varying mean-reversion parameter $\mu(t)$ demonstrates a clear increasing trend, rising from about $0.97$ to nearly 1. The results are reported in \Cref{table:2024}, which presents the final point estimates of the parameters. 

\begin{table}[htbp]
    \caption{Parameter estimates for the 2024 U.S.\ presidential election}
    \label{table:2024}
    \centering
    \begin{tabular}{ccccc}
    \hline
    Parameter & Point estimate & S.E. & 95\% CI lower & 95\% CI upper \\
    \hline
    $\theta$
    & 0.181
    & 0.0108
    & 0.159
    & 0.202 \\
    \hline
    $\sigma^2$
    & $4.15 \times 10^{-6}$
    & $2.91 \times 10^{-7}$
    & $3.62 \times 10^{-6}$
    & $4.76 \times 10^{-6}$ \\
    \hline
    $\omega^2$
    & $2.35 \times 10^{-6}$
    & $2.27 \times 10^{-7}$
    & $1.91 \times 10^{-6}$
    & $2.80 \times 10^{-6}$ \\
    \hline
    \end{tabular}
\end{table}

\subsubsection{Model evaluation}
We apply the Bradley-Terry model trained on the 2020 U.S.\ presidential election data to select the candidate-specific odds quote at each signal time. We use a simplified odds-price representation in which decimal odds are treated as price values. The selected odds $O_{i,t}$ are used as the entry value, and $O_{i,t+k}$ is used as the exit value. For the $j^{\text{th}}$ trade, we define the realized odds-price return as
\begin{align*}
R_{j,t\to t+k}=\frac{O_{i_j,t+k}-O_{i_j,t}}{O_{i_j,t}},
\end{align*}
where $O_{i_j,t}$ and $O_{i_j,t+k}$ denote the entry and exit decimal odds, respectively. Under this price-based representation, the trade generates a positive return when
$O_{i_j,t+k}>O_{i_j,t}$. Thus, $R_{j,t\to t+k}$ measures the proportional change in the selected decimal-odds quote between entry and exit.

We summarize the profitability of the strategy using the sample mean of the realized odds-price returns,
\begin{align*}
\bar{R}=\frac{1}{N}\sum_{j=1}^{N}R_j,
\end{align*}
where $N$ is the number of completed trades generated by the trading signals. The quantity $\bar{R}$ is the mean realized return across trades, rather than the theoretical expected return.

We assess return variability using the sample standard deviation
\begin{align*}
s_R=\sqrt{\frac{1}{N-1}\sum_{i=1}^{N}
\left(R_j-\bar{R}\right)^2
}.
\end{align*}

We then report the unannualized per-trade Sharpe-type ratio
\begin{align*}
SR=\frac{\bar{R}-R_f}{s_R},
\end{align*}
where $R_f$ denotes the risk-free return over the holding period. Because the trades are short-lived and the corresponding risk-free return is negligible, we set $R_f=0$, giving
\begin{align*}
SR=\frac{\bar{R}}{s_R}.
\end{align*}
This statistic measures the mean realized odds-price return relative to its cross-trade variability. Because the trades may have unequal holding periods and the statistic is not annualized, we interpret it as a descriptive per-trade Sharpe-type ratio rather than a conventional annualized portfolio Sharpe ratio \citep{sharpe1998sharpe}. Within the peer-to-peer gambling framework, these descriptive statistics summarize the magnitude and cross-trade variability of the synthetic odds-price returns generated by the strategy. After calculation, the 2024 dataset yielded a positive sample mean odds-price return of 1.86\% and an unannualized per-trade Sharpe-type ratio of 1.12 under the simplified frictionless representation. The out-of-sample performance supports the use of the Bradley-Terry model for candidate selection in the 2024 U.S.\ presidential election and indicates its potential applicability to other electoral markets dominated by two leading candidates.

\subsubsection{Model diagnostics}

We conducted the diagnostic checks summarized in \ref{Model Diagnostic}. For the trending Ornstein-Uhlenbeck model, the Q-Q plot of the standardized residuals shows a systematic departure from the normal reference line, indicating that the Gaussian assumption does not fully describe the residual distribution. The residual ACF also reveals substantial negative autocorrelation at lag 1, positive autocorrelation at lag 2, and several smaller correlations outside the approximate 95\% confidence bounds. Thus, the model captures the broad trend and mean-reverting structure of the data, but does not completely account for its short-range dependence or conditional distribution.

For the Bradley-Terry model, the Hosmer-Lemeshow-type test produced a $p$-value greater than 0.05, providing insufficient evidence to reject adequate calibration across the probability groups considered. Overall, the models provide useful approximations to the main features of the data, although the remaining non-Gaussianity and serial dependence should be taken into account when interpreting the prediction intervals and trading results.

\section{Conclusion}\label{sec4}

This study demonstrates the feasibility of adapting a pairs-trading framework to political prediction markets. We model the combined implied probability of the two major-party nominees using a latent time-varying-mean Ornstein-Uhlenbeck process with additive observation noise, use bootstrap-based prediction bounds to identify trading signals, and apply a Bradley-Terry pairwise comparison model to select the candidate-specific odds. In the 2024 out-of-sample evaluation, the strategy generated 130 completed trades, with a mean realized odds-price return of 1.86\% and an unannualized per-trade Sharpe-type ratio of 1.12. These results illustrate the potential of the integrated signal-generation and candidate-selection framework in electoral markets dominated by two leading candidates.

However, the models have some limitations in this paper. First, the models do not account for order queues in real betting exchanges. Other bettors may already have unmatched orders at the same price level, so
a newly submitted order may not be matched immediately. As a result, the odds observed at the signal time may no longer be available when the order reaches the front of the queue. Second, the analysis treats genuinely unavailable observations as missing completely at random (MCAR). This assumption may be restrictive if missing records are associated with periods of high market activity or technical disruption. In such cases, the missing-data mechanism could affect model estimation and prediction. Third, the strategy uses an event-driven exit rule under which a position remains open until the first subsequently observed change in the selected quote. Consequently, holding periods vary across trades and depend on market activity. Future work should examine the distribution and predictability of these quote durations, as well as the sensitivity of the results to alternative fixed-horizon and event-driven exit rules. Fourth, the reported returns measure changes in decimal odds rather than actual betting profits. They do not account for transaction costs or execution conditions in a real betting exchange. Furthermore, computational speed must be improved, as signals can change within minutes in real-time markets and the transaction fee should also be considered. A rolling estimation window may reduce computational cost and improve the timeliness of signal generation. Finally, a simple pairs trading strategy is limited to scenarios with two main outcomes; in multi-party elections, such as in Germany, a basket of correlated odds may be more appropriate. In addition, since the prediction interval is 95\%, altering it to 99\% could lead to different betting decisions. Despite these limitations, the proposed framework provides a useful basis for modeling political-odds dynamics and developing signal-based trading rules, while further applications are needed to evaluate its robustness across different elections and market settings.

\section*{Ethics} This work did not require ethical approval from a human subject or animal welfare committee.

\section*{Data accessibility} The gambling data used in this study, and all R code to produce the results, are available on GitHub at \url{https://github.com/lhy199661/2024-US-Election}. The repository contains the Betfair 2020 and 2024 U.S.\ Election data, which were purchased from the company BetData.  Redistribution of these data is permitted under the purchase agreement.   
BetData collected all data from Betfair using an API. 

\section*{Declaration of AI use} 
The authors used ChatGPT (OpenAI GPT-5.5 in Chat and Codex, OpenAI GPT Sol 5.6)
for language editing and methodological and statistical development. 
For language editing, the LLM was used to locate grammar errors, detect 
deviations from academic writing style, and to suggest updates which
were sometimes implemented exactly and sometimes implemented with edits. 
An initial draft of the statistical procedure described in this work and the code to implement the method was produced without
LLM use. Several refinements and corrections were made by or motivated by LLM output. Bugs in code were located and 
corrected with LLMs. For the method, the only prominent example of 
an update triggered by LLM interaction is our use of the Kalman filter. In the pre-LLM draft of the paper, we erroneously
treated the noised OU process as Markov which led to an
incorrect likelihood.

\section*{Authors' contributions} 
H.L.: data curation, investigation, methodology, validation, visualization, writing-original draft, writing-review and editing; 
L.T.: data curation, methodology, project administration, supervision, writing-review and editing; 
B.R.B.: data curation, methodology, supervision, writing-review and editing; 
C.D.: conceptualization, supervision, writing-review and editing.\\
All authors gave final approval for publication and agreed to be held accountable for the work performed therein.

\section*{Conflict of interest declaration} The authors declare no competing interests.

\section*{Funding} Haoyu Liu gratefully acknowledges funding from the China Scholarship Council.

\section*{Acknowledgements}
The authors thank Valentin Popov for valuable and constructive feedback.

\appendix 
\renewcommand{\thesection}{Appendix \Alph{section}}
\renewcommand{\thesubsection}{\Alph{section}.\arabic{subsection}}
\renewcommand{\theequation}{\Alph{section}.\arabic{equation}}

\section{Time-varying mean-reversion level}\label{trending OU with noise}

For the full-market process $Y_t$ defined in the \cref{fair-market benchmark}, the observed over-round is
\begin{align*}
    \Omega_t^{\mathrm{obs}}=Y_t-1.
\end{align*}
The observed over-round may change over time as market conditions and odds change \citep{williams1999information}. We distinguish this
observed quantity from its equilibrium level, denoted by $\overline{\Omega}(t)$.

The corresponding equilibrium level of the full-market inverse-odds sum is
\begin{align*}
    \mu_Y(t)=1+\overline{\Omega}(t).
\end{align*}
The fair-market benchmark is obtained when $\overline{\Omega}(t)=0$, giving an equilibrium level of one. A constant positive over-round gives a constant equilibrium above one. If $\overline{\Omega}(t)$ changes over time, then the equilibrium level also changes over time.

We therefore model the latent full-market process as
\begin{align*}
    \d P_t
    &=
    \theta
    \left\{
        \mu_Y(t)-P_t
    \right\}
    \d t
    +
    \sigma\,\d W_t, \\
    Y_t
    &=
    P_t+E_t.
\end{align*}
This specification reduces to a standard Ornstein-Uhlenbeck process when $\mu_Y(t)$ is constant. A time-varying $\overline{\Omega}(t)$ instead produces a time-varying mean-reversion level.

The empirical analysis considers only the two major-party nominees. Define their combined inverse odds as
\begin{align*}
    Z_t
    =
    \frac{1}{O_{1,t}}
    +
    \frac{1}{O_{2,t}},
\end{align*}
and let
\begin{align}
    Z_t^{*}
    =
    \sum_{i=3}^{n}\frac{1}{O_{i,t}}
\end{align}
denote the contribution of the remaining outcomes. By construction,
\begin{align*}
    Y_t=Z_t+Z_t^{*}.
\end{align*}

For the following derivation, we approximate the contribution of the remaining outcomes by a deterministic continuously differentiable function,
\begin{align*}
    Z_t^{*}=g(t).
\end{align*}
This assumption simplifies the derivation. A stochastic specification for the remaining outcomes would require an additional model for their drift, variation, and dependence with the two-candidate process.

Define
\begin{align*}
    Q_t=P_t-g(t).
\end{align*}
Since $Z_t=Y_t-g(t)$, the corresponding observation equation is
\begin{align*}
    Z_t=Q_t+E_t.
\end{align*}

Using
\begin{align*}
    \d Q_t=\d P_t-g'(t)\d t
\end{align*}
and $P_t=Q_t+g(t)$ gives
\begin{align*}
    \d Q_t
    &=
    \theta
    \left\{
        \mu_Y(t)-Q_t-g(t)
    \right\}
    \d t
    -
    g'(t)\d t
    +
    \sigma\,\d W_t \notag\\
    &=
    \theta
    \left\{
        \mu_Z(t)-Q_t
    \right\}
    \d t
    +
    \sigma\,\d W_t,
\end{align*}
where
\begin{align*}
    \mu_Z(t)
    =
    \mu_Y(t)
    -
    g(t)
    -
    \frac{g'(t)}{\theta}.
\end{align*}

Substituting $\mu_Y(t)=1+\overline{\Omega}(t)$ gives
\begin{align*}
    \mu_Z(t)
    =
    1+\overline{\Omega}(t)
    -
    g(t)
    -
    \frac{g'(t)}{\theta}.
\end{align*}
Hence,
\begin{align*}
    \mu_Z'(t)
    =
    \overline{\Omega}'(t)
    -
    g'(t)
    -
    \frac{g''(t)}{\theta}.
\end{align*}

The mean-reversion level of the two-candidate process is therefore time varying whenever
\begin{align*}
    \overline{\Omega}'(t)
    -
    g'(t)
    -
    \frac{g''(t)}{\theta}
    \not\equiv 0.
\end{align*}
Changes in the market over-round and in the inverse-odds contribution of the remaining outcomes can therefore change the equilibrium level of $Z_t$. A positive over-round affects the level of the process, but it does not by itself imply a time trend. The direction and magnitude of the trend are estimated from the data.

\section{Conditional moments}\label{One-step Conditional Moments}
\subsection{Markov property}\label{Markov Property}
\begin{lem}
\label{lem:markov}
   The stochastic process $P_t$ satisfies the Markov property.
\end{lem}

The Markov property is a standard property of the Ornstein-Uhlenbeck process \citep{oksendal2003stochastic}. For completeness, we verify this property directly from the transition representation below.

\begin{proof}
    Since the Ornstein-Uhlenbeck process can be regarded as the continuous-time analogue of the discrete-time AR(1) process \citep{tang2009parameter}, its exact discrete-time transition can be obtained at equally spaced observation times. Assuming that the interval between two consecutive observations is one unit of model time, $\Delta t=1$, the transition equation is
   \begin{equation*}
    P_t=e^{-\theta}P_{t-1}+\mu\left(1-e^{-\theta}\right)
    +\sigma\sqrt{\frac{1-e^{-2\theta}}{2\theta}}
    \,\varepsilon_t,
   \end{equation*}
  where $\varepsilon_t\overset{\mathrm{iid}}{\sim} \mathcal{N}(0,1)$ is a standard normal innovation independent of the process history up to time $t-1$. Thus we have
    \begin{align*}
    \delta P_t &=P_{t}-P_{t-1} \\
    &=P_{t-1}(e^{-\theta}-1)+\mu(1-e^{-\theta})+\sigma\sqrt{\frac{1-e^{-2\theta}}{2\theta}}z_{t-1}.
\end{align*}
    We conclude that the conditional distribution of $P_t$ only depends on $P_{t-1}$ rather than $P_1,\cdots,P_{t-2}$. Thus $P_t$ satisfies the Markov property.
\end{proof}

\begin{lem}
\label{lem:markov2}
   The stochastic process $Q_t$ satisfies the Markov property.
\end{lem}

\begin{proof}
    For any $s<t$, the solution of the trending OU process is
\begin{align*}
Q_t
&=
e^{-\theta(t-s)}Q_s
+
\theta
\int_s^t
e^{-\theta(t-u)}
\mu(u)\,\d u\\
&\quad+
\sigma
\int_s^t
e^{-\theta(t-u)}
\,\d W_u.
\end{align*}
The second term is deterministic, while the stochastic integral depends only on the Wiener-process increment over $(s,t]$ and is independent of the process history up to time $s$. Consequently, conditional on $Q_s$, the distribution of $Q_t$ is independent of $\{Q_u:u<s\}$. Hence, $\{Q_t\}$ is a time-inhomogeneous Markov process.
\end{proof}

\subsection{Kalman filter for trending OU process with additive white noise}\label{parameter for trending OU with additive noise}
 An Ornstein-Uhlenbeck and trending OU process with additive observation noise do not satisfy the Markov property when only the observed process is considered. The observed series represents a noisy measurement of an underlying mean-reverting latent state, and the presence of measurement noise obscures the true state dynamics. Therefore, the conditional distribution of the next observation depends not only on the current observation but also on previous observations that contain information about the latent state. Consequently, the observed process alone is not Markov.

 Nevertheless, the system can be represented as a linear Gaussian state-space model in which the latent OU state evolves according to a Markov process and the observed variable is a noisy measurement of that state. The joint dynamics are therefore described through a state equation and an observation equation \citep{kalman1960new}.

 Conditional on a correctly specified linear Gaussian state-space model and known parameter values, the Kalman filter recursively computes the exact conditional mean and covariance of the latent state given the available observations \citep{van2004sigma}. Under these assumptions, the filtered conditional mean is the minimum mean-square-error estimator of the latent state. In the empirical implementation, the unknown parameters are replaced by their estimated values, producing plug-in filtered and predicted state estimates.

 The resulting state estimates, prediction-error covariances, and innovations are used for likelihood-based parameter estimation and one-step-ahead forecasting \citep{durrant2001introduction}. The Kalman filter thus provides a recursive method for accounting for both latent state uncertainty and additive observation noise within the assumed state-space specification \citep{bao2024application}.

Trending OU process with additive white noise $\{Z_t\}$ defined by 
\begin{align*}
    \d Q_t 
    & = \theta \{\mu(t)-Q_t\} \, \d t + \sigma \d W_t, \\
    Z_t & = Q_t + E_t, \nonumber
\end{align*}
where $\mu(\cdot)$ is a deterministic, potentially time-varying function of time $t$, and the other terms are defined as in previous part. Define $M\left(t\right)=\mathbb{E}\left(Z_t\right)=\mathbb{E}\left(Q_t\right)$. According to state-space form, the linear Gaussian state-space model is specified as:
\begin{align}
\text{State equation }: &\quad Q_t = a_t Q_{t-1} + b_t + \eta_t, \label{eq:state_ou} \\
\text{Observation equation}: &\quad Z_t = Q_t + E_t, \label{eq:obs_ou}
\end{align}
where $Q_t$ is unobserved latent state (trending OU process), $a_t = e^{-\theta}$, $\theta$ is mean reversion rate, $b_t=\hat M(t)-a_t \hat M(t-1)$ is time-varying offset, $\hat M^{\prime}(t)=\theta\{\mu(t)-\hat M(t)\}$ which comes from smoothed trends, $\eta_t \sim \mathcal{N}(0, q_t)$ which is state noise, $q_t = \frac{\sigma^2}{2\theta}(1 - e^{-2\theta})$, $\sigma$ is state noise volatility, $Z_t$ is observed noisy measurement, $E_t \sim \mathcal{N}(0, \omega^2)$ is the market microstructure noise, and $\omega$ is standard deviation of the market microstructure noise.

The Kalman filter produces two critical quantities for likelihood calculation at each time step $t$, \citep{kalman1960new}:
\begin{itemize}
    \item Innovation : $\tilde{z}_t = Z_t - \hat{Q}_{t|t-1}$, where $\hat{Q}_{t|t-1}$ is the predicted state estimate from the Kalman filter.
    \item Innovation covariance: $S_t = P_{t|t-1} + \omega^2$, where $P_{t|t-1}$ is the predicted state error covariance.
\end{itemize}

For the trending OU with additive noise model, the Kalman filter recursions required for computing $\hat{Q}_{t|t-1}$ and $P_{t|t-1}$ is:

\begin{align}
\hat{Q}_{t|t-1} &= a_t \hat{Q}_{t-1|t-1} + b_t, \notag \\ 
P_{t|t-1} &= a_t^2 P_{t-1|t-1} + q_t, \notag \\ 
\tilde{Z}_t &= Z_t - \hat{Q}_{t|t-1},  \notag \\ 
S_t &= P_{t|t-1} + \omega^2, \label{innovation covariance} \\ 
K_t &= \frac{P_{t|t-1}}{S_t},  \notag \\ 
\hat{Q}_{t|t} &= \hat{Q}_{t|t-1} + K_t \tilde{Z}_t,  \notag \\ 
P_{t|t} &= (1 - K_t) P_{t|t-1}. \notag
\end{align}

Accordingly, after incorporating the observations $Z_1,\ldots,Z_t$, the one-step-ahead predictive distribution is
\begin{align*}
Z_{t+1}\mid Z_1,\ldots,Z_t
\sim
\mathcal{N}\left(
\hat Q_{t+1\mid t},
S_{t+1}
\right),
\quad
S_{t+1}
=
P_{t+1\mid t}+\omega^2.
\end{align*}

We initialize the latent state prior as
\begin{align*}
    Q_0
    \sim
    N\left(
        M(0),
        \frac{\sigma^2}{2\theta}
    \right),
\end{align*}
so that
\begin{align*}
    \hat Q_{0|0}
    &=
    M(0),\\
    P_{0|0}
    &=
    \frac{\sigma^2}{2\theta}.
\end{align*}
The first observation is then incorporated using the standard measurement update \citep{bao2024application}.

A linear Gaussian state-space model implies that the conditional distribution of $Z_t$, given the previous observations $Z_1,\ldots,Z_{t-1}$, is Gaussian. Its mean is the predicted observation $\hat Q_{t\mid t-1}$, and its variance is the innovation variance $S_t$,
\begin{align}
    \mathbb{P}(z_t \mid \mathcal{Z}_{t-1}) = \mathcal{N}\left(\tilde{Z}_t; 0, S_t\right) = \frac{1}{\sqrt{2\pi S_t}} \exp\left(-\frac{\tilde{Z}_t^2}{2S_t}\right). \label{eq:c.10}
\end{align}

This result can be attributed to the fact that
\begin{itemize}
    \item The innovation $\tilde{z}_t = Z_t - \hat{Q}_{t|t-1}$ is a linear combination of Gaussian random variables \citep{kalman1960new}.
    \item By construction of the Kalman filter's a priori estimate $\mathbb{E}[\tilde{z}_t \mid \mathcal{Z}_{t-1}] = 0$.
    \item From ~\cref{innovation covariance}, $\text{Var}(\tilde{z}_t \mid \mathcal{Z}_{t-1}) = S_t$.
\end{itemize}

We use the chain rule of probability, and find the joint likelihood of the full observation sequence $\mathcal{Z} = \{Z_1, Z_2, ..., Z_n\}$ where $n$ = total number of observations is the product of the conditional likelihoods:
\begin{align}
    \mathbb{P}(\mathcal{Z} \mid \theta, \sigma^2, \omega^2) = \prod_{t=1}^n \mathbb{P}(z_t \mid \mathcal{Z}_{t-1}, \theta, \sigma^2, \omega^2), \label{conditional observations}
\end{align}
where the likelihood is conditioned on the parameters of $\theta, \sigma^2, \omega^2$ because the Kalman filter outputs $(\hat{Q}_{t|t-1}, P_{t|t-1}, S_t)$ depend on them.

We convert the \cref{conditional observations}  to a sum:
\begin{align}
    \log \mathbb{P}(\mathcal{Z} \mid \theta, \sigma^2, \omega^2) = \sum_{t=1}^n \log \mathbb{P}(z_t \mid \mathcal{Z}_{t-1}, \theta, \sigma^2, \omega^2). \label{log equation}
\end{align}

We substitute the Gaussian conditional likelihood \cref{eq:c.10} into \cref{log equation}:
\begin{align}
    \log \mathbb{P}(\mathcal{Z} \mid \theta, \sigma^2, \omega^2) = \sum_{t=1}^n \left[ -\frac{1}{2}\log(2\pi) - \frac{1}{2}\log(S_t) - \frac{\tilde{z}_t^2}{2S_t} \right]. \label{log equation2}
\end{align}

For the trending OU with additive noise model, we substitute the Kalman filter quantities into \cref{log equation2} to get the final log-likelihood expression tailored to the trending OU process with additive noise:
\begin{align*}
    L(\theta, \sigma^2, \omega^2; \mathcal{Z}) = -\frac{n}{2}\log(2\pi) - \frac{1}{2}\sum_{t=1}^n \left[ \log\left(P_{t|t-1} + \omega^2\right) + \frac{\left(z_t - \left(a_t \hat{Q}_{t-1|t-1} + b_t\right)\right)^2}{P_{t|t-1} + \omega^2} \right].
\end{align*}

For two successive retained observations, let $q_j$ denote the number of original observation intervals between them. Taking one observation interval as one unit of time, the exact state transition is
\begin{align*}
    Q_{\tau_j}
    &=
    a_jQ_{\tau_{j-1}}
    +
    M(\tau_j)
    -
    a_jM(\tau_{j-1})
    +
    \eta_j,\\
    a_j
    &=
    \exp(-\theta q_j),\\
    \eta_j
    &\sim
    N\left(
        0,
        \frac{\sigma^2}{2\theta}
        (1-a_j^2)
    \right).
\end{align*}

The data are sampled at regular observation times and may include periods with no change in $Z_t$. We treat consecutive identical values as a single retained observation for likelihood estimation. The time spanned by the omitted repetitions is still preserved in the model through the interval between retained observations.

Let $\tau_j$ denote the original observation index of the $j^\text{th}$ retained value. We define
\begin{align*}
    \widetilde Z_j = Z_{\tau_j},
\quad j=1,\ldots,N,
\end{align*}
and
\begin{align*}
    q_j=\tau_j-\tau_{j-1}, \quad j=2,\ldots,N.
\end{align*}

Here, $q_j$ gives the number of original observation intervals between two successive retained values. A value of $q_j>1$ may indicate that the process remained unchanged at intermediate observation points or that one or more scheduled observations were unavailable. We take one observation interval as one unit of time. The interval $q_j$ is therefore used directly in the continuous-time OU state transition.

At each retained observation index, the observation equation is
\begin{align*}
    \widetilde Z_j
    =
    Q_{\tau_j}+E_{\tau_j},
    \quad
    E_{\tau_j}
    \overset{\mathrm{iid}}{\sim}
    N(0,\omega^2),
    \quad
    j=1,\ldots,N,
\end{align*}
where $Q_{\tau_j}$ denotes the latent trending Ornstein-Uhlenbeck process and $E_{\tau_j}$ denotes additive observation noise.

Over the interval $(\tau_{j-1},\tau_j]$, the exact transition of the latent process is
\begin{align*}
    Q_{\tau_j}
    &=
    e^{-\theta q_j}Q_{\tau_{j-1}}
    +
    M(\tau_j)
    -
    e^{-\theta q_j}M(\tau_{j-1})
    +
    \eta_j,
\end{align*}
where
\begin{align*}
    \eta_j
    =
    \sigma
    \int_{\tau_{j-1}}^{\tau_j}
    e^{-\theta(\tau_j-s)}
    \,\mathrm dW_s.
\end{align*}
The state-transition coefficient is therefore
\begin{align*}
    a_j    = e^{-\theta q_j}.
\end{align*}

By It\^o isometry,
\begin{align*}
    \operatorname{Var}(\eta_j)
    &=
    \sigma^2
    \int_{\tau_{j-1}}^{\tau_j}
    e^{-2\theta(\tau_j-s)}
    \,\mathrm ds\\
    &=
    \sigma^2
    \int_0^{q_j}
    e^{-2\theta u}
    \,\mathrm du\\
    &=
    \frac{\sigma^2}{2\theta}
    \left(
        1-e^{-2\theta q_j}
    \right).
\end{align*}
Hence, the process-noise variance is
\begin{align*}
    v_j=
    \frac{\sigma^2}{2\theta}
    \left(
        1-e^{-2\theta q_j}
    \right).
\end{align*}

Indexing the Kalman filter by the retained observations, the prediction recursions are
\begin{align*}
    \hat Q_{j\mid j-1}
    &=
    a_j\hat Q_{j-1\mid j-1}
    +
    \hat M(\tau_j)
    -
    a_j\hat M(\tau_{j-1}),\\
    P_{j\mid j-1}
    &=
    a_j^2P_{j-1\mid j-1}
    +
    v_j.
\end{align*}
The innovation and its conditional variance are
\begin{align*}
    \nu_j
    &=
    \widetilde Z_j-\hat Q_{j\mid j-1},\\
    S_j
    &=
    P_{j\mid j-1}+\omega^2.
\end{align*}
The measurement-update recursions are
\begin{align*}
    K_j
    &=
    \frac{P_{j\mid j-1}}{S_j},\\
    \hat Q_{j\mid j}
    &=
    \hat Q_{j\mid j-1}
    +
    K_j\nu_j,\\
    P_{j\mid j}
    &=
    (1-K_j)P_{j\mid j-1}.
\end{align*}
Thus, unequal intervals between successive retained observations are incorporated through the interval-specific quantities $a_j$ and $v_j$.

The marginal mean function $\hat M(\cdot)$ is first estimated using the smoothing procedure described in the main text. Conditional on this estimated function, the Gaussian innovation log-likelihood is
\begin{align*}
    L\left(
        \theta,\sigma^2,\omega^2;
        \hat M
    \right)
    =
    -\frac{1}{2}
    \sum_{j=1}^{N}
    \left[
        \log(2\pi)
        +
        \log S_j
        +
        \frac{\nu_j^2}{S_j}
    \right].
\end{align*}
The structural parameters are then estimated by
\begin{align*}
    \left(
        \hat\theta,
        \hat\sigma^2,
        \hat\omega^2
    \right)'
    =
    \underset{
        \theta,\sigma^2,\omega^2
    }{\arg\max}
    \;
    L\left(
        \theta,\sigma^2,\omega^2;
        \hat M(\cdot)
    \right) \quad \text{s.t.} \quad \sigma^2 > 0, \omega^2 > 0. 
\end{align*}

Finally, the estimated time-varying mean-reversion level is obtained as
\begin{align*}
    \hat\mu(t)
    =
    \hat M(t)
    +
    \frac{\hat M'(t)}{\hat\theta}.
\end{align*}

\section{Data simulation}\label{Data Simulation}
\subsection{Simulation design}
To validate the theoretical derivations and the Kalman filter–based log-likelihood from \cref{eq7:MLE}, we conduct a Monte Carlo study under a trending OU process with additive white noise. For given true parameters $(\theta, \sigma^2, \omega^2)$, we generate 1000 observations per replication. The first 500 observations are used as the training sample for parameter estimation via L-BFGS maximization of the derived log-likelihood, while the remaining 500 observations are reserved for out-of-sample evaluation. This procedure is repeated 1000 times to assess finite-sample performance. The marginal mean function was specified as $M(t)=5+0.001t+0.5\sin\left(\frac{2\pi t}{200}\right), \quad t=1,\ldots,1000,$ and the generating mean-reversion level was defined by $\mu(t)=M(t)+\frac{M'(t)}{\theta},$ where $M'(t)$ was approximated numerically using finite differences. The
initial latent state was fixed at $Q_1=M(1)$. The parameter values were $\theta=0.15$, $\sigma^2=0.55$, and $\omega^2=0.20$, and the random seed used for data generation was $34351$.

\subsection{Parameter distribution and unbiasedness}
\begin{figure}[H]
    \centering
    \includegraphics[width=3in]{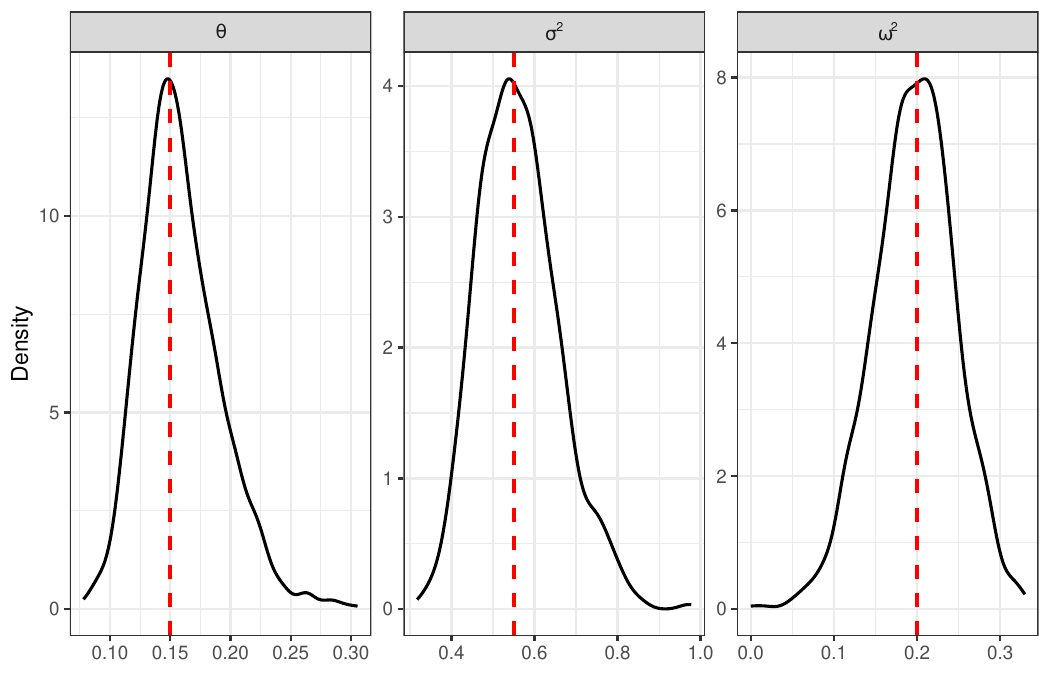}
    \captionsetup{justification=centering}
    \caption{Sampling distributions of the estimated parameters across the simulation replications. The three panels show the distributions of $\theta$, $\sigma^2$, and $\omega^2$, respectively. The black curves represent the empirical density of the parameter estimates, and the red dashed lines indicate the corresponding generating parameter values.}
    \label{figure:distribution}
\end{figure}
Across replications, the sampling distributions are approximately unimodal and centred near the generating parameter values in the simulation settings considered from \Cref{figure:distribution}. 

\begin{figure}[H]
    \centering
    \includegraphics[width=3in]{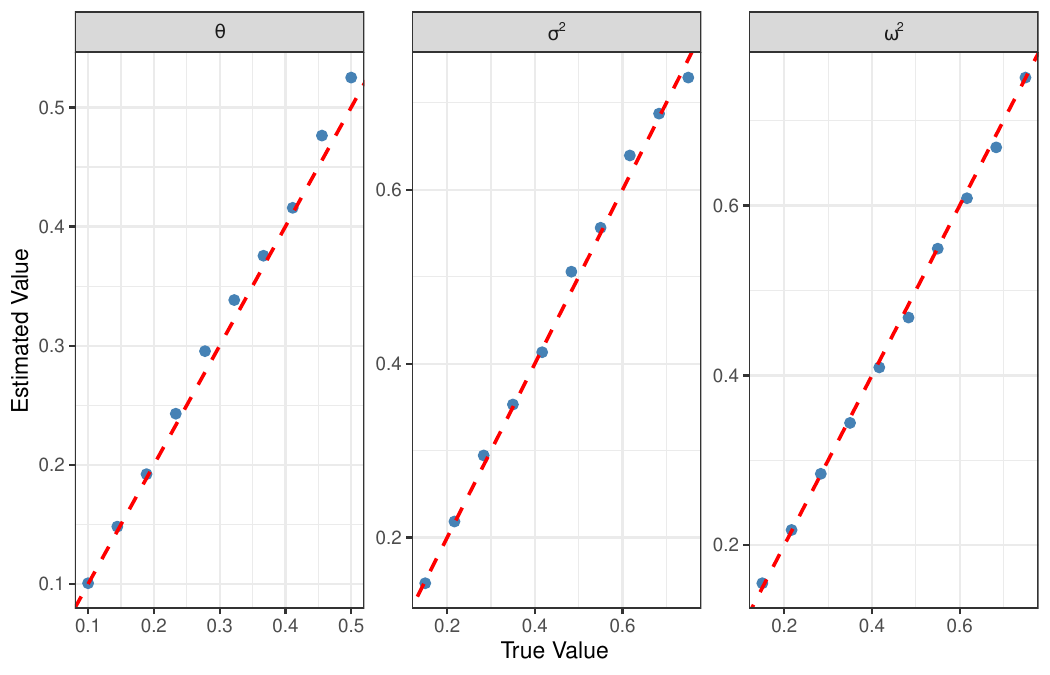}
    \captionsetup{justification=centering}
    \caption{Estimated versus generating parameter values in the simulation study. The three panels correspond to $\theta$, $\sigma^2$, and $\omega^2$, respectively. The blue points show the estimated values for each generating parameter value, and the red dashed lines represent the identity line. Estimates close to this line indicate agreement between the estimated and generating values.}
    \label{figure:unbiased}
\end{figure}
To further examine unbiasedness, we vary $(\theta, \sigma^2, \omega^2)$ over pre-specified grids and compare the estimated means with their corresponding true values. The results from \Cref{figure:unbiased} show no systematic deviation from the 45-degree line, confirming that the estimators are approximately centred around the generating values in the simulation settings considered.

\subsection{Predictive coverage}
\begin{figure}[H]
    \centering
    \includegraphics[width=4in]{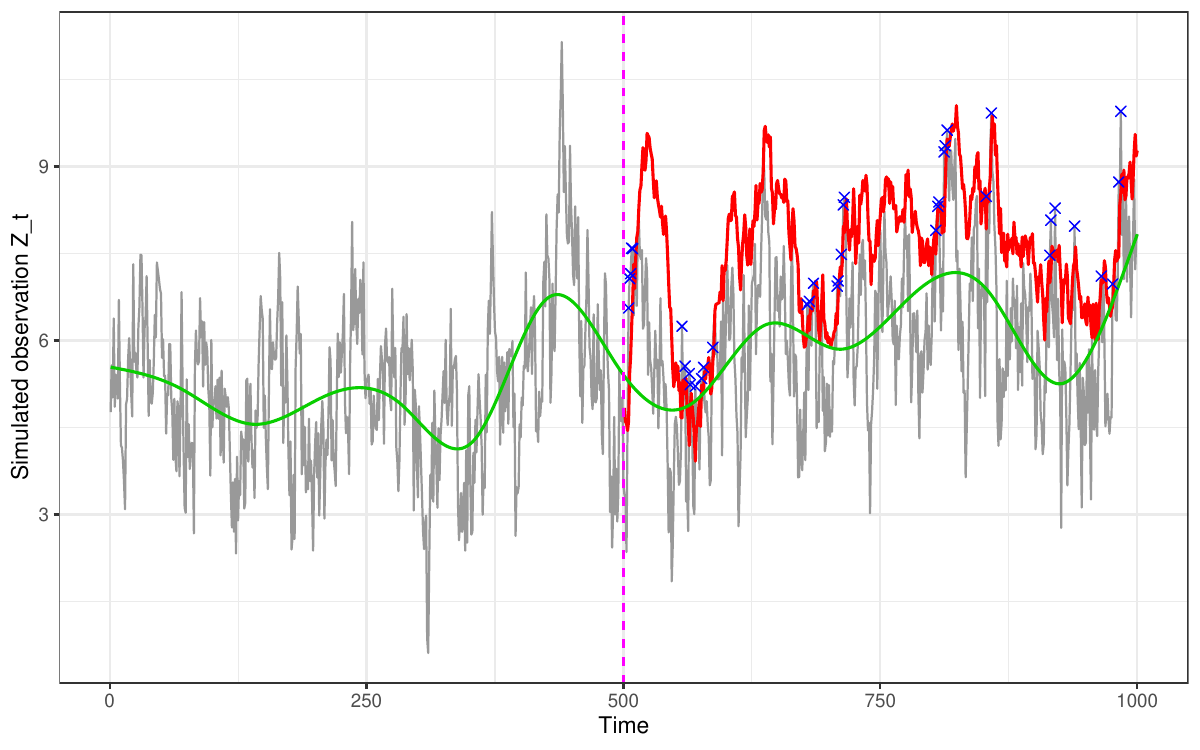}
    \captionsetup{justification=centering}
    \caption{Data Simulation. The  blue vertical lines are signal points, green line is time-dependent mean-reversion level $\mu(\cdot)$, red line is the upper one-sided prediction bound, and purple dashed line is the starting prediction points.}
    \label{figure:simulate}
\end{figure}
Out-of-sample predictive performance is evaluated using one-step-ahead prediction intervals from \cref{eq9:prediction bound}. In a representative simulation from \Cref{figure:simulate}, 25 out of 500 test observations fall outside the 95\% one-sided prediction interval, consistent with nominal coverage. Aggregating over 1000 replications, the empirical coverage probability equals 94.8\%, which closely matches the theoretical 95\% level. This agreement indicates that, under the simulated data-generating equation, the implemented state-space recursions, likelihood evaluation, and filtering procedure behave consistently with the model used to generate the synthetic data.

\section{Forecast calibration and candidate selection}\label{Model Selection}
\subsection{When to place a bet}\label{When}
To evaluate the calibration of the proposed prediction intervals, we compute
\begin{align*}
    \hat{F}_{t\mid t-1}(z_t) = \frac{1}{B}\sum_{b=1}^B I\!\left(Z_t^{*(b)} \leq z_t\right), 
\end{align*}
where $\hat{F}_{t\mid t-1}(z_t)$ represents the bootstrap-estimated probability that the realized value $z_t$ lies within the simulated predictive distribution.
The upper one-sided prediction bound is then defined as
\begin{align*}
    U_t = \min \{ z : \hat{\mathbb{P}}(z) \geq 1-\alpha \}. 
\end{align*}
\begin{figure}[H]
    \centering
    \includegraphics[width=5in]{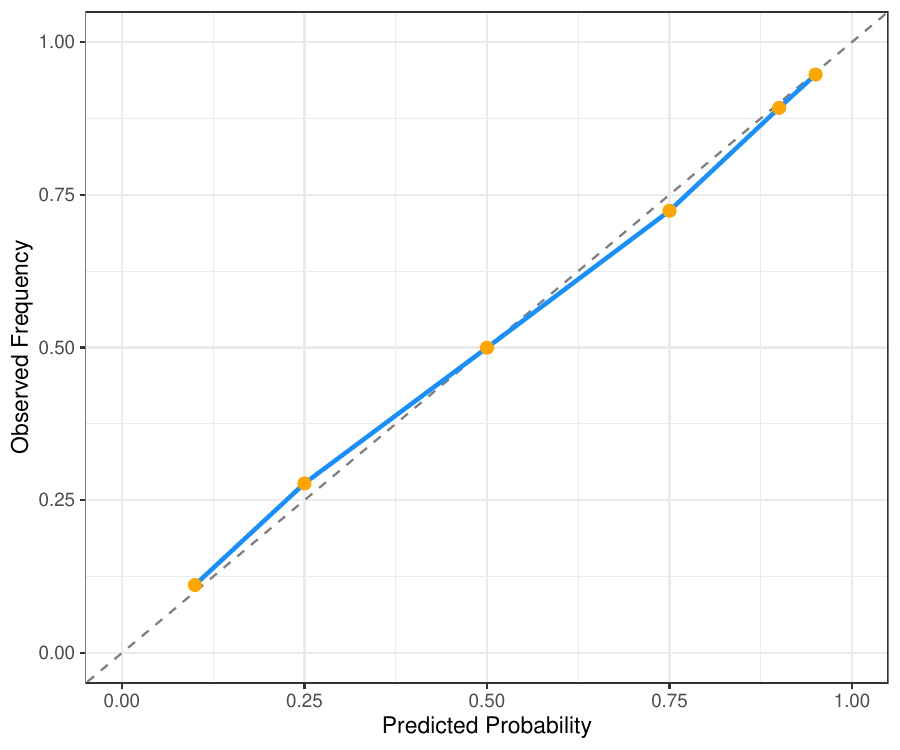}
    \captionsetup{justification=centering}
    \caption{Reliability diagram comparing predicted probabilities with observed frequencies. The orange points show the empirical frequencies within probability groups, and the blue line connects these points. The grey dashed diagonal line represents perfect calibration, where the observed frequency equals the predicted probability.}
    \label{figure:reliability}
\end{figure}


The empirical coverage rate is 95.2\%, which is close to the nominal coverage level of 95\%. In addition, the reliability diagram plots the empirical coverage against nominal quantiles \citep{brocker2007increasing}. \Cref{figure:reliability} illustrates the calibration plot comparing predicted probabilities with observed frequencies. The dashed diagonal line represents perfect calibration, where predicted probabilities match empirical event rates. The orange curve shows the realized calibration performance. Overall, the model demonstrates good calibration, as the observed frequencies closely track the diagonal, particularly in the mid to high probability ranges. Minor deviations are present at lower probabilities, where the model tends to be slightly under-confident. Nonetheless, the results suggest that the probabilistic forecasts capture the empirical distribution reasonably well, supporting their reliability for inference and decision-making.

\subsection{Which bet to place}\label{which}
We use the Bradley-Terry model to determine which candidate-specific odds to select. Let
\begin{align*}
    X_{i,t}
    =
    \begin{pmatrix}
        Y_{i,t}\\
        \Delta Y_{i,t}
    \end{pmatrix},
    \quad
    \Delta Y_{i,t}
    =
    Y_{i,t}-Y_{i,t-1},
    \quad
    i\in\{1,2\},
\end{align*}
where $i=1$ denotes the Republican nominee and $i=2$ denotes the Democratic nominee. The conditional probability that candidate 1 is preferred to candidate 2 is
\begin{align*}
    \pi_t
    &:=
    \mathbb{P}
    \left(
        L_t=1
        \mid
        X_{1,t},X_{2,t}
    \right)\\
    &=
    \frac{1}{
        1+
        \exp\left[
            -\boldsymbol{\beta}^{\top}
            \left(
                X_{1,t}-X_{2,t}
            \right)
        \right]
    }.
\end{align*}
The estimated coefficients for the 2020 U.S.\ presidential election are reported in \Cref{table:4}. No intercept is included in the pairwise logistic model, thereby preserving antisymmetry under an exchange of the two candidates. The Republican nominee's odds are selected when $\pi_t\geq0.5$, whereas the Democratic nominee's odds are selected when $\pi_t<0.5$.

\begin{table}[htbp]
    \caption{Coefficients of the 2020 Bradley-Terry model}
    \label{table:4}
    \centering
    \begin{tabular}{ccccc}
    \hline
    Parameters & Estimate & Std. Error & $z$ value & $Pr(>|z|)$ \\
    \hline
    $Y_{1,t}-Y_{2,t}$ 
    & -0.150 
    & 0.789 
    & -0.190 
    & 0.849 \\
    \hline
    $\Delta Y_{1,t}-\Delta Y_{2,t}$ 
    & 33.1 
    & 23.0 
    & 1.44 
    & 0.150 \\
    \hline    
    \end{tabular}
\end{table}

\section{Model diagnostics}\label{Model Diagnostic}
The outputs of trending OU process with additive white noise model diagnostics show
\begin{figure}[H]
  \centering
  \begin{subfigure}[b]{0.45\textwidth}
 \includegraphics[width=2.5in]{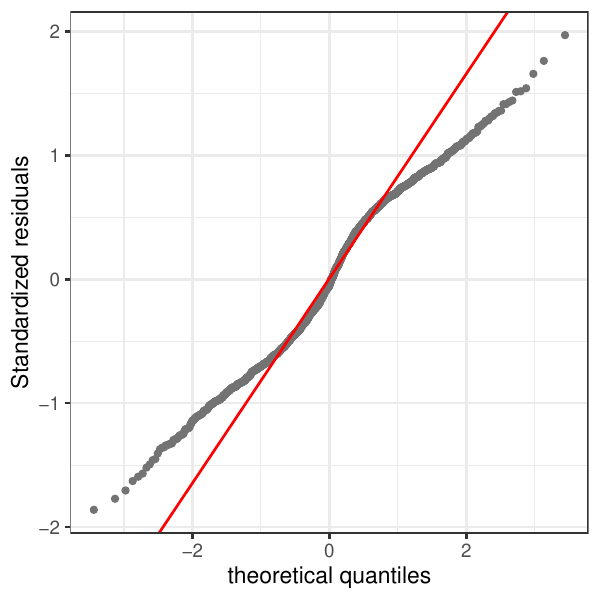}
    \captionsetup{justification=centering}
    \caption{2024 U.S.\ Presidential Election Q-Q Plot}
    \label{fig:QQ1}
  \end{subfigure}
  \hfill
  \begin{subfigure}[b]{0.45\textwidth}
\includegraphics[width=2.5in]{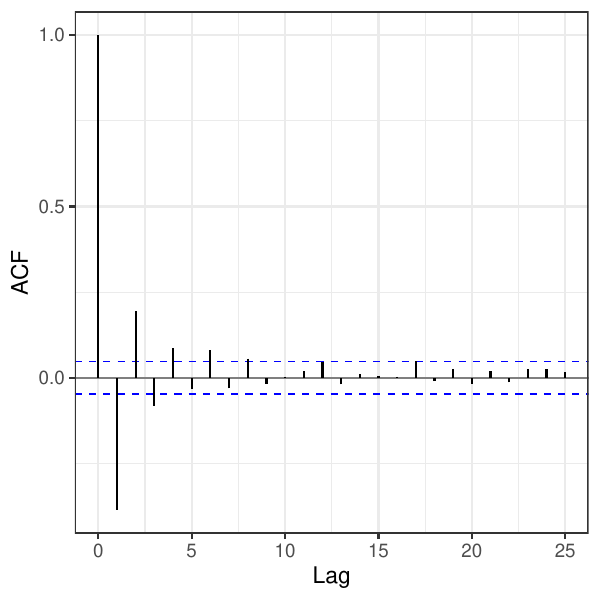}
    \captionsetup{justification=centering}
    \caption{2024 U.S.\ Presidential Election ACF Plot}
    \label{fig:ACF1}
  \end{subfigure}
  \captionsetup{justification=centering}
  \caption{Model diagnostics for the 2024 U.S.\ presidential election. Panel (a) shows the normal Q-Q plot of the standardized residuals, with the red line indicating the normal-reference line. Panel (b) shows the sample autocorrelation function (ACF) of the standardized residuals. The blue dashed lines indicate the approximate 95\% bounds for zero autocorrelation.}
  \label{fig:2024}
\end{figure}

\begin{figure}[H]
  \centering
  \begin{subfigure}[b]{0.45\textwidth}
    \includegraphics[width=2.5in]{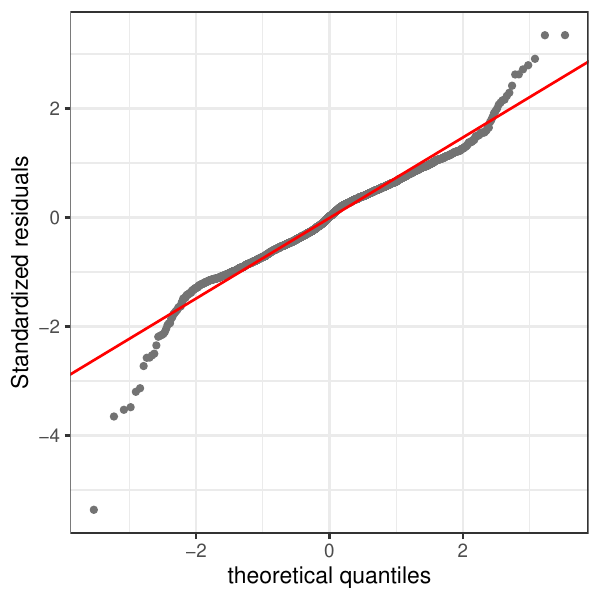}
    \captionsetup{justification=centering}
    \caption{2020 U.S.\ Presidential Election Q-Q Plot}
    \label{fig:QQ2}
  \end{subfigure}
  \hfill
  \begin{subfigure}[b]{0.45\textwidth}
\includegraphics[width=2.5in]{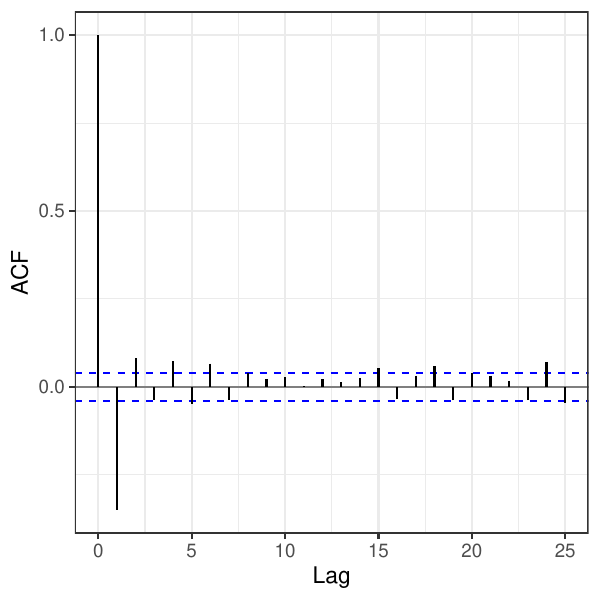}
    \captionsetup{justification=centering}
    \caption{2020 U.S.\ Presidential Election ACF Plot}
    \label{fig:ACF2}
  \end{subfigure}
  \captionsetup{justification=centering}
  \caption{Model diagnostics for the 2020 U.S.\ presidential election. Panel (a) shows the normal Q-Q plot of the standardized residuals, with the red line indicating the normal-reference line. Panel (b) shows the sample autocorrelation function (ACF) of the standardized residuals. The blue dashed lines indicate the approximate 95\% bounds for zero autocorrelation.}
  \label{fig:2020}
\end{figure}

For the Bradley-Terry model, the Hosmer-Lemeshow-type test yielded $\chi^2=9.58$ with eight degrees of freedom and a $p$-value of 0.296, providing insufficient evidence to reject adequate calibration across the groups considered.


\newpage
\begin{spacing}{0.1}

\bibliography{sn-bibliography}
\end{spacing}

\end{document}